\documentclass[10pt,a4paper,dvipsnames]{article}

\usepackage{fullpage,amssymb,amsmath,amsfonts,hyperref,amsthm,paralist,graphicx,color,float, mathtools,tikz}
\usepackage[utf8]{inputenc}
\usepackage[english]{babel}
\usepackage{here}
\usepackage[normalem]{ulem}
\usepackage{centernot}
\usepackage{ifthen,xkeyval, tikz, calc, graphicx}
\usepackage{xcolor}
\usepackage{framed}
\usepackage[textsize=scriptsize, backgroundcolor=blue!20!white,bordercolor=red]{todonotes}
\usepackage{dsfont}
\usepackage{mathrsfs}
\usepackage{comment}
\usepackage{stmaryrd}
\usepackage{amssymb}
\usepackage{xpatch}

\newtheorem{theorem}{Theorem}[section]
\newtheorem{lemma}[theorem]{Lemma}
\newtheorem{prop}[theorem]{Proposition}

\newtheorem{observation}[theorem]{Observation}

\theoremstyle{definition}
\newtheorem{definition}[theorem]{Definition}
\newtheorem*{assumption}{Assumption}

\makeatletter
\xpatchcmd{\assumption}{\@addpunct{.}}{\normalfont\,\@addpunct{:}}{}{}
\makeatother

\definecolor{shadecolor}{named}{GreenYellow}

\newcommand{\al}[1]{\begin{align*}#1\end{align*}}

\newcommand{\algn}[1]{\begin{align}#1\end{align}}
\newcommand{\eqq}[1]{\begin{equation}#1\end{equation}}

\newcommand{\p}{\mathbb P}
\newcommand{\pla}{\mathbb P_\lambda}
\newcommand{\E}{\mathbb E}
\newcommand{\ela}{\E_\lambda}
\newcommand{\R}{\mathbb R}
\newcommand{\Rd}{\mathbb R^d}
\newcommand{\Zd}{\mathbb Z^d}

\newcommand{\N}{\mathbb N}

\newcommand{\dd}{\, \mathrm{d}}
\newcommand{\C}{\mathscr {C}}

\newcommand{\piv}[1]{\textsf {Piv}(#1)}
\newcommand{\conn}[3]{#1 \longleftrightarrow #2\textrm { in } #3}
\newcommand{\nconn}[3]{#1 \centernot\longleftrightarrow #2\textrm { in } #3}
\newcommand{\dconn}[3]{#1 \Longleftrightarrow #2\textrm { in } #3}
\newcommand{\xconn}[4]{#1 \xleftrightarrow{\,\,#4\,\,} #2\textrm { in } #3}

\newcommand{\thinning}[2]{#1_{\langle #2 \rangle}}
\newcommand{\tlam}{\tau_\lambda}
\newcommand{\dtlam}{\sigma_\lambda}
\newcommand{\dcf}{g_\lambda}

\newcommand{\orig}{\mathbf{0}}
\newcommand{\e}{\text{e}}

\newcommand{\connf}{\varphi}
\newcommand{\rg}{\Gamma_\connf}

\newcommand{\weight}{\textbf{w}}
\newcommand{\ord}{\textsf{ord}}
\newcommand{\gen}[1]{\langle\!\langle #1 \rangle\!\rangle}

\newcommand{\twodoin}[1]{\todo[inline, backgroundcolor=green!20!white,bordercolor=red]{#1}}

\numberwithin{equation}{section}
\allowdisplaybreaks

\title{The Direct-Connectedness Function in the Random Connection Model}
\author{Sabine Jansen \and Leonid Kolesnikov \and Kilian Matzke}

\begin{document}
\maketitle

\begin{abstract}
We investigate expansions for connectedness functions in the random connection model of continuum percolation in powers of the intensity. Precisely, we study the pair-connectedness and the direct-connectedness functions, related to each other via the Ornstein-Zernike equation. We exhibit the fact that the coefficients of the expansions consist of sums over connected and $2$-connected graphs. In the physics literature, this is known to be the case more generally for percolation models based on Gibbs point processes and stands in analogy to the formalism developed for correlation functions in liquid-state statistical mechanics.

We find a representation of the direct-connectedness function and bounds on the intensity which allow us to pass to the thermodynamic limit. In some cases (e.g., in high dimensions), the results are valid in almost the entire subcritical regime. Moreover, we relate these expansions to the physics literature and we show how they coincide with the expression provided by the lace expansion.\\

\noindent \emph{Keywords}: Ornstein-Zernike equation, random connection model, connectedness functions, Poisson process, percolation, graphical expansions, lace expansion.\\

\noindent \emph{MSC 2020}:   60K35, 82B43, 60G55, 60D05.

\end{abstract}



\section{Introduction and main result} \label{sec:intro}

Perturbation analysis plays an important role in both  stochastic geometry \cite[Chapter 19]{LasPen17} and statistical mechanics. For Gibbs point processes (grand-canonical Gibbs measures in statistical mechanics), quantities like factorial moment densities (also called correlation functions) are highly non-trivial functions of the intensity of the Gibbs point process itself (density) or the intensity of 
an underlying Poisson point process (activity).  When interactions are pairwise, it is is well-known that the coefficients of these expansions are given by sums over geometric, weighted graphs. There is a vast literature addressing the convergence of these expansions, see, for example, \cite{brydges1986leshouches,malyshev-minlos}. Some attempts have been made at exploiting power series expansions from statistical mechanics for likelihood analysis of spatial point patterns in spatial statistics, see \cite{OgaTan89}.

The physics literature provides similar power series expansions for connectedness functions in a class of percolation models driven by Gibbs point processes, the so-called random connection models (RCM)~\cite{ConDeAngFor77}. The expansion coefficients for the pair connectedness function can be written in terms of a sum of certain connected graphs (see~\eqref{eq:tlam_cluster_exp}) and the coefficients for the direct-connectedness function in terms of a sum over certain $2$-connected graphs (see~\eqref{eq:dcf_rewriting}). The two functions relate via the Ornstein-Zernike equation (OZE)~\cite{OrnZer14}, an integral equation which is of paramount importance in physical chemistry and soft matter physics and which enters some approaches to percolation theory, see \cite[Chapter 10]{torquato-book}. For Bernoulli bond percolation on $\mathbb Z^d$, the Ornstein-Zernike equation encodes a renewal structure and it is used to prove Ornstein-Zernike behavior~\cite{campanino-ioffe2002}, a precise asymptotic formula for pair connectedness functions in the subcritical regime that incorporates subleading corrections to the exponential decay. The Ornstein-Zernike equation also appears as a by-product of lace expansions \cite[Proposition~5.2]{HeyHofLasMat19}.

The expansions for connectedness functions appearing in \cite{ConDeAngFor77} are derived as a means to discuss the following question: Is it possible to choose the notion of connectivity in such a way that the percolation transition, if it occurs at all, coincides with the phase transition in the sense of non-uniqueness of Gibbs measures? We remind the reader that the relation between the two phenomena is rather subtle and in general, the corresponding critical parameters do not match, see~\cite{jansen2016percolation} and references therein. To the best of our knowledge,  the question above has not been fully answered for continuum systems, although Betsch and Last~\cite{Betsch2021OnTU} were recently able to show that uniqueness of the Gibbs measure follows from the non-percolation of an associated RCM driven by a Poisson point process. 

Moreover, the convergence of the expansions for connectedness functions has not been treated in a mathematically rigorous way, in stark contrast with the rich theory of cluster expansions. Even in the simplest case of the RCM driven by a Poisson point process that we consider in this paper, where activity and density coincide and are called the intensity, rigorous results for the expansion of connectedness functions barely exist: The first ones were obtained by Last and Ziesche in~\cite{LasZie17}. However Last and Ziesche do not prove that their expansions coincide with the physicists' expansion, and they do not prove quantitative bounds for the domain of convergence of the small-intensity expansion.

Our main result addresses graphical expansions of the direct-connectedness function in infinite volume. The results by Last and Ziesche~\cite{LasZie17}, combined with our combinatorial considerations from Section~\ref{sec:connections_lz}, imply that the physicists' expansions have a positive radius of convergence, however it is not our purpose to provide a quantitative bound for the latter. Instead, we perform first a resummation, in finite volume, of the physicists' expansion. Although the resummed expansion is no longer a power series in the intensity of the underlying Poisson point process, it has the (conjectured) advantage of converging in a bigger domain than the physicists' expansion.  We provide quantitative bounds on the intensity that allow us to pass to the infinite-volume limit in the resummed expansion of the direct-connectedness function. The proof uses the continuum BK inequality proved in~\cite{HeyHofLasMat19}.

In addition, we discuss the relation of the physicists' and our expansion with the lace expansion for the continuum random connection model \cite{HeyHofLasMat19}. Roughly, the lace expansion could in theory be rederived from the graphical expansion by yet another resummation step. In fact a notion of laces similar to the laces for  for the self-avoiding random walk~\cite{brydges1986leshouches,Sla06} already enters the proof of our main result on graphical expansions (see Section~\ref{sec:laces}). Thus, contrary to what is stated in \cite[Chapter~6.1]{HeyHof17} the denomination ``lace expansion'' for percolation is not a misnomer, at least for continuum systems. It is unclear, however, whether the discussion  offers a new angle of attack on the intricate convergence problems in the theory of lace expansions. \\

Let us properly introduce the RCM and state our results. The RCM depends on two parameters, namely its \emph{intensity} $\lambda \geq 0$ and the (measurable) \emph{connection function} $\connf\colon \Rd \to [0,1]$, satisfying
	\eqq{\nonumber0<\int \connf(x) \dd x <\infty}
as well as radial symmetry $\connf(x) = \connf(-x)$ for all $x\in\Rd$. The model is described informally as follows: The vertex set is taken to be a homogeneous Poisson point process (PPP) in $\Rd$ of intensity $\lambda$, denoted by $\eta$. For any pair $x,y\in\eta$, we add the edge $\{x,y\}$ with probability $\connf(x-y)$ and independently of all other pairs. We refer to~\cite{HeyHofLasMat19, MeeRoy96} for a formal construction.

The RCM is an undirected simple random spatial graph and a standard model of continuum percolation. We denote it by $\xi$ and we use $\pla$ to denote the according probability measure. Its vertex set is $V(\xi)=\eta$, and we let $E(\xi)$ denote its edge set.

For $x\in\Rd$, we let $\xi^{x}$ be the RCM augmented by the point $x$. In other words, the vertex set of $\xi^x$ is $\eta \cup \{x\}$ and the edges are formed as described above. In particular, edges between $x$ and points of $\eta$ are drawn independently and according to $\connf$. More generally, for a set of points $x_1, \ldots, x_k$, we let $\xi^{x_1, \ldots, x_k}$ be the RCM with vertex set $\eta\cup \{x_1, \ldots, x_k\}$ (also here, edges between deterministic points $x_1, x_2$ are drawn independently and according to $\connf$). 

We say that $x,y \in \eta$ are connected (and write $\conn{x}{y}{\xi}$) if there is a path from $x$ to $y$ in $\xi$. For $x\in\Rd$, we let $\C(x) = \C(x,\xi^x) = \{ y \in \eta^x: \conn{x}{y}{\xi^x} \}$ be the cluster of $x$ and define the \emph{pair-connectedness} (or \emph{two-point}) function $\tlam\colon\Rd\times\Rd \to [0,1]$ to be
	\eqq{ \tlam(x,y) := \pla\big( \conn{x}{y}{\xi^{x,y}}\big).   \label{eq:def:tlam}}
Due to the translation invariance of the model, we have $\tlam(x,y) = \tlam(\orig,x-y)$ (where $\orig$ denotes the origin in $\Rd$) and we can also define $\tlam$ as a function $\tlam\colon\Rd \to [0,1]$ with $\tlam(x) = \pla(\conn{\orig}{x}{\xi^{\orig,x}})$.

We say that $x,y\in\eta$ are \emph{$2$-connected} (or \emph{doubly connected}) and write $\dconn{x}{y}{\xi}$ if there are two paths from $x$ to $y$ that have only their endpoints in common (or if $x$ and $y$ are directly connected by an edge or if $x=y$). We define 
	\[\dtlam(x) := \pla(\dconn{\orig}{x}{\xi^{\orig,x}}). \]
  Recall that the critical intensity for percolation is defined by
	\[\lambda_c = \sup \{\lambda \geq 0: \pla(|\C(\orig)| = \infty)=0 \}\]
and the identity
\[\sup \{\lambda \geq 0: \pla(|\C(\orig)| = \infty)=0 \}=\sup\{ \lambda \geq 0: \int \tlam(x) \dd x < \infty\}\]
was shown to hold true for connectivity functions $\phi$ that are non-increasing in the Euclidean distance (see~\cite{Mee95}).
It is proved in~\cite{LasZie17} that for $\lambda<\lambda_c$, there exists a uniquely defined integrable and essentially bounded function $\dcf:\R^d\times\R^d\to \R^d$ such that
	\eqq{  \tlam(x,y) = \dcf(x,y) + \lambda\int_{\R^d}  \dcf(x,z)\tlam(z,y)\dd(z),\quad x,y\in\R^d. \label{eq:intro:oze_2d}} 
This equation is known as the Ornstein-Zernike equation (OZE) and $\dcf$ is called the \emph{direct-connectedness function}.

For two integrable functions $f,g\colon \Rd \to \R$, we recall the convolution $f\ast g$ to be given by
	\[ (f \ast g)(x) = \int_{\Rd} f(x) g(x-y) \dd y.\]
We let $f^{\ast 1} = f$ and $f^{\ast m} = f^{\ast (m-1)} \ast f$. Notice that we can interpret both the pair-connectedness function $\tau_\lambda$ and the direct-connectedness function  $\dcf$ as functions on $\R^d$ due to translation invariance. The OZE then can be formulated as
\eqq{  \tlam = \dcf + \lambda ( \dcf\ast\tlam). \label{eq:intro:oze}}

Naturally, the question whether one can provide an explicit form for the direct-connectedness function $\dcf$ arises. Unfortunately, an immediate probabilistic interpretation of $\dcf$ is not known. One classical approach from the physics literature is to obtain explicit approximations for the solution $g_\lambda$ of~\eqref{eq:intro:oze_2d} by introducing complementary equations, so called closure relations, the choice of which depends on the specifics of the model considered. Different closure relations provide different explicit approximations for $\dcf$ and thus also for the pair-connectedness function $\tau_\lambda$, e.g., via a reformulation of the OZE~\eqref{eq:intro:oze_2d} for the Fourier transforms of the connectedness functions. Most prominent are the Percus-Yevick closure relations~\cite{torquato-book,ChiSte89}, other examples can be found in~\cite{hansen2013theory}. Another approach~\cite{ConDeAngFor77} is to directly provide an independent definition of $\dcf$ in terms of a graphical expansion and then argue that this expansion satisfies the OZE~\eqref{eq:intro:oze_2d}. We follow the spirit of the latter approach: Our main result is a graphical  expansion for the direct-connectedness function, with quantitative bounds on the domain of convergence.

Let
	\eqq{ \lambda_\ast := \sup \Big\{ \lambda \geq 0: \sup_{x \in \Rd} \sum_{k \geq 1} \lambda^{k-1} \dtlam^{\ast k}(x) < \infty\Big\}, 
			\qquad \tilde\lambda_\ast := \sup\Big\{ \lambda \geq 0: \lambda \int \dtlam(x) \dd x < 1\Big\}. \label{eq:def:lambda_ast}}
It is not hard to see that $\tilde\lambda_\ast \leq \lambda_\ast \leq \lambda_c$ using   \eqref{eq:lambda_ast_bound} below.

We can now state our main theorem. It provides (in general dimension) the first rigorous quantitative bounds on $\lambda$ under which the direct-connectedness function admits a convergent graphical expansion.

\begin{theorem}[Graphical expansion of the direct-connectedness function] \label{thm:main_thm} \
For $\lambda<\lambda_\ast$,  the direct-connectedness function $\dcf(x_1,x_2)$ is given by the expansion~\eqref{eq:gfinal} which is absolutely convergent pointwise for all $(x_1,x_2)\in \mathbb{R}^{2d}$. 
Moreover, for $\lambda< \tilde \lambda_\ast$, the expansion~\eqref{eq:gfinal} converges in the $L^1(\R^d,\mathrm d x_2)$-norm for all $x_1\in \mathbb R^d$. 
\end{theorem}
The convergence results for the expansion~\eqref{eq:gfinal} are proved in Theorem~\ref{thm:dcf_inf:limit_theorem} and Theorem~\ref{thm:dcf_inf:limit_theorem:L1}, the equality with the direct-connectedness function is proved in Section~\ref{sec:oze}.

Last and Ziesche show that there is some $\lambda_0>0$ so that $\dcf$ is given by a power series for $\lambda \in [0,\lambda_0)$. No quantitative bounds for $\lambda_0$ are provided however. In Section~\ref{sec:connections_lz}, we discuss how to relate this expansion to our expression for $\dcf$. We now make several remarks on Theorem~\ref{thm:main_thm} and the quantitative nature of the bounds provided there. 
\begin{enumerate}
\item[$\bullet$] Since $0 \leq \dtlam\leq 1$, we can bound
	\eqq{ \sum_{k \geq 1} \lambda^{k-1} \dtlam^{\ast k}(x) \leq \sum_{k \geq 0} \Big( \lambda \int \dtlam(x) \dd x \Big)^k
		= \sum_{k \geq 0} \big( \ela\big[|\{x \in \eta: \dconn{\orig}{x}{\xi^{\orig}}\}| \big] \big)^k, \label{eq:lambda_ast_bound}}
where the identity is due to the Mecke equation~\eqref{eq:prelim:mecke_m}. This shows that $\tilde\lambda_\ast \leq \lambda_\ast$ and that $\tilde\lambda_\ast$ is the point where the expected number of points in $\eta$ that are $2$-connected to the origin passes $1$ (i.e., we have $\ela\big[|\{x \in \eta: \dconn{\orig}{x}{\xi^{\orig}}\}| \big]\geq 1$ for all $\lambda>\tilde{\lambda}_\ast$).
\item[$\bullet$] The argument of the geometric series in~\eqref{eq:lambda_ast_bound} can be further bounded from above by
	\[ \lambda \int \tlam(x) \dd x =  \ela[|\{ x \in \eta: \conn{\orig}{x}{\xi^{\orig}}\}|] ,\]
the expected cluster size (minus $1$). A classical branching-process argument gives that $\tilde\lambda_\ast \geq 1/2$  (see, for example,~\cite[Theorem 3]{Pen93}).
\item[$\bullet$] In high dimension, we have the following result, proven in~\cite{HeyHofLasMat19}: Under some additional assumptions on $\connf$ (see~\cite[Section~1.2]{HeyHofLasMat19}), there is an absolute constant $c_0$ such that
	\[ \lambda_c \int \sigma_{\lambda_c}(x) \dd x \leq 1 + c_0/d \]
in sufficiently high dimension, or, for a class of \emph{spread-out} models (closely related to Kac potentials in statistical mechanics, see~\cite{HarSla90}) with a parameter $L$,
	\[\lambda_c\int \sigma_{\lambda_c}(x) \dd x \leq 1 + c_0 L^{-d}\]
for all dimensions $d>6$ (in the spread-out case, $c_0$ is independent of $L$ but may depend on $d$). As $\sigma_\lambda$ is non-decreasing in $\lambda$, this provides a bound for the whole subcritical regime. This also implies that for every $\varepsilon>0$, there is $d_0$ (respectively, $L_0$) such that $\tilde\lambda_\ast \geq 1-\varepsilon$ for all $d \geq d_0$ (respectively, $L \geq L_0$ and $d>6$). As we also know that $\lambda_c \searrow 1$ as the dimension becomes large, this shows that in high dimension, $\tilde\lambda_\ast$ (and thus also $\lambda_\ast$) gets arbitrarily close to $\lambda_c$.
\end{enumerate}

\paragraph{Outline of the paper.} The paper proceeds as follows. We introduce most of our important notation in Section~\ref{sec:defs}. This allows us to demonstrate some basic (and mostly well-known) central ideas in Section~\ref{sec:twopoint_fct}, where the two-point function is discussed in finite volume. Section~\ref{sec:dcfshell} contains the main body of work for the proof of Theorem~\ref{thm:main_thm} (the convergence results). The remainder of Theorem~\ref{thm:main_thm} regarding the OZE is then proved in Section~\ref{sec:oze}.

We discuss our results in Section~\ref{sec:connections}. In particular, we point out where many of the formulas can be found in the physics literature (not rigorously proven) and allude to generalizations to Gibbs point processes. Moreover, we highlight the connection to two other expressions for the pair connectedness function; in particular, we show how our expansions relate to the lace expansion. Lastly, we address other percolation models very briefly in Section~\ref{sec:bondperco}.

\section{Fixing notation} \label{sec:defs}
\subsection{General notation}
We let $[n] := \{1,\ldots, n\}$ and $[n]_0 := [n] \cup \{0\}$. For a set $V$, we write $\binom{V}{2} := \{E \subseteq V: |E|=2\}$. For $I = \{i_1, i_2, \ldots, i_\kappa\} \subset \N$, let $\vec x_I=(x_{i_1}, \ldots, x_{i_\kappa})$. For compact intervals $[a,b]\subset \R$, we write $\vec x_{[a,b]} = \vec x_I$ with $I=[a,b]\cap \N$. If $a=1$, we write $\vec x_{[b]} = \vec x_{[1,b]}$. By some abuse of notation, we are going to interpret $\vec x_{[a,b]}$ both as an ordered vector and as a set.

If not specified otherwise, $\Lambda$ denotes a bounded, measurable subset of $\Rd$.

\subsection{Graph theory}
We recall that a (simple) graph $G=(V,E) = (V(G), E(G))$ is a tuple with \emph{vertex} set $V$ (or set of \emph{points, sites, nodes}) and \emph{edge} set $E \subseteq \binom{V}{2}$ (or set of \emph{bonds}). In this paper, we will always consider graphs with $V \subset \Rd$, and for $x,y\in\Rd$, an edge $\{x,y\}$ will sometimes be abbreviated $xy$.

If $xy\in E$, we write $x \sim y$ (and say that $x$ and $y$ are \emph{adjacent}). We extend this notation and write $x \sim W$ for $x \in V$ and $W \subseteq V$ if there is $y \in W$ such that $x \sim y$; also, write $A \sim B$ if there is $x\in A$ such that $x \sim B$. For $W \subseteq V$, we define the $W$-\emph{neighborhood} $N_W(x) = \{y \in W: x \sim y\}$ and the $W$-\emph{degree} of a vertex $x\in V$ as $\deg_W(x) = |N_W(x)|$, and we write $N(x) = N_V(x)$ as well as $\deg(x) = \deg_V(x)$. For two sets $A,B \subseteq V$, write $E(A,B) = \{xy \in E(G): x\in A, y \in B\}$.

Given a graph $G=(V,E)$ and $W \subseteq V$, we denote by $G[W] := (W, \{e \in E: e \subseteq W\})$ the subgraph of $G$ \emph{induced} by $W$. Given two simple graphs $G,H$, we let $G \oplus H := (V(G) \cup V(H), E(G) \cup E(H))$.

\paragraph{Connectivity.} Given a graph $G$ and two of its vertices $x,y\in V(G)$, we say that $x$ and $y$ are \emph{connected} if there is a \emph{path} between $x$ and $y$---that is, a sequence of vertices $x=v_0, v_1, \ldots, v_k=y$ for some $k\in\N_0$ such that $v_{i-1}v_i \in E(G)$ for $i \in [k]$. We write $x \longleftrightarrow y$ in $G$ or simply $x \longleftrightarrow y$. We call $\C(x)=\C(x;G) = \{y \in V(G): x \longleftrightarrow y\}$ the \emph{cluster} (or \emph{connected component)} of $x$ in $G$. If there is only one cluster in $G$, we say that $G$ is connected.

For $x \longleftrightarrow y$ in $G$, we let $\piv{x,y;G}$ denote the set of \emph{pivotal} vertices for the connection between $x$ and $y$ . That is, $v\notin \{x,y\}$ is in $\piv{x,y;G}$ if every path from $x$ to $y$ in $G$ passes through $v$. We say that $x$ is doubly connected to $y$ in $G$ (and write $\dconn{x}{y}{G}$) if $\piv{x,y;G}=\varnothing$.  We remark that in the physics literature, pivotal points are usually known as \emph{nodal} points.

In the pathological case $x=y$, we use the convention $x \longleftrightarrow x$ in $G$ and set $\piv{x,x;G}=\varnothing$ for any graph $G$ with $x\in V(G)$ (equivalently, $\dconn{x}{x}{G}$).

We observe that the pivotal points $\{u_1,\ldots, u_k\}$ can be ordered in a way such that every path from $x$ to $y$ passes through the pivotal points in the order $(u_1, \ldots, u_k)$. We define $\textsf{PD}(x,y,G) = \textsf{PD}(G)$ to be the \emph{pivot decomposition} of $G$, that is, a partition of the vertex set $V$ into a sequence, $(x,V_0, u_1, V_1, \ldots, u_k, V_k, y)$, where $(u_1, \ldots, u_k)$ are the ordered pivotal points and $V_i$ is the (possibly empty) set of vertices that can be reached only by passing through $u_i$ and that is still connected to $x$ after removing $u_{i+1}$. See Figure~\ref{fig:core_pivot_decomp}.

\begin{figure}
	\centering
 \includegraphics[scale=0.9]{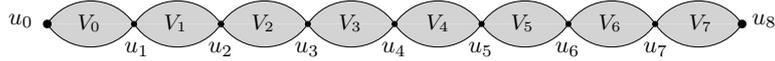}
	\caption{A schematic sketch of the pivot decomposition $(u_0,V_0, \ldots, V_7, u_8)$ of $G$, setting $x=u_0$ and $y=u_{k+1}$.}
	\label{fig:core_pivot_decomp}
\end{figure}

\paragraph{Classes of graphs.} Given a (locally finite) set $X \subset \Rd$, we let $\mathcal G(X)$ be the set of graphs with vertex set $X$. We let $\mathcal C(X)$ be the set of connected graphs on $X$. Moreover, for $x,y \in X$, we let $\mathcal D_{x,y}(X)\subseteq \mathcal C(X)$ be the set of non-pivotal graphs, i.e.,~the set of connected graphs such that $\piv{x,y;G} = \varnothing$.

Given $m$ \emph{bags} $X_1, \ldots, X_m \subset \Rd$ with $|X_i \cap X_j| \leq 1$ for all $1 \leq i<j \leq m$ , we let $\mathcal G(X_1,\ldots, X_m)$ denote the set of \emph{$m$-partite} graphs on $X_1, \ldots, X_m$, i.e.,~the set of graphs $G$ with $V(G) = \cup_{i=1}^m X_i$ and $E(G[X_i]) =\varnothing$ for $i\in[m]$. Note that we allow bags to have (at most) one vertex in common, which is a slight abuse of the notation in graph theory, where $m$-partite graphs have disjoint bags.

\paragraph{The notion of ($\pm$)-graphs.} We introduce a ($\pm$)-graph as a triple $G^\pm = (V(G),E^+(G), E^-(G)) = (V,E^+,E^-)$, where $V$ is the vertex set and $E^+,E^- \subseteq \binom{V}{2}$ are disjoint. In other words $G^\pm$ is a graph where every edge is of exactly one of two types (plus or minus). We set $E := E^+ \cup E^-$ and associate to $G^\pm$ the two simple graphs $G^{|\pm|} := (V,E)$ and $G^+:=(V^+,E^+)$, where $V^+ := \{x \in V: \exists e \in E^+: x \in e\}$ are the vertices incident to at least one $(+)$-edge.

We extend all the notions for simple graphs to ($\pm$)-graphs. In particular, given $X \subset \Rd$, we let $\mathcal G^\pm(X)$ be the set of ($\pm$)-graphs on $X$. Moreover, $\mathcal C^\pm(X)$ are the ($\pm$)-connected graphs on $X$, that is, the graphs such that $G^{|\pm|}$ is connected. Similarly, $\mathcal C^+(X)\subset \mathcal C^\pm(X)$ are the $(+)$-connected graphs, that is, those where $G^+$ is connected and $V(G)=V^+$. For $x,y\in X$, we denote by $\mathcal D^\pm_{x,y}(X)$ the set of those ($\pm$)-connected graphs on $X$ where $\piv{x,y;G^{|\pm|}}= \varnothing$ and $\mathcal D^+_{x,y}(X)\subset \mathcal D^\pm_{x,y}(X)$ are those ($\pm$)-connected graphs on $X$ where $\piv{x,y;G^+}=\varnothing$. We also define the $(\pm)$-pivot decomposition $\textsf{PD}^\pm(x,y,G^\pm) = \textsf{PD}^\pm(G^\pm) = \textsf{PD}(G^{|\pm|})$ and the $(+)$-pivot decomposition $\textsf{PD}^+(x,y,G^\pm) = \textsf{PD}^+(G^\pm) = \textsf{PD}(G^{+})$. Lastly, we write $x \overset{+}{\longleftrightarrow} y$ if there is a path from $x$ to $y$ in $E^+$.

Given a ($\pm$)-graph $G$ and a simple graph $H$, we define $G \oplus H := (V(G) \cup V(H), E^+(G), E^-(G) \cup E(H))$.

\paragraph{Weights.} Given a simple graph $G$, a ($\pm$)-graph $H$ on $X \subset \Rd$ and the connection function $\connf$, we define the weights
	\[ \weight (G) := (-1)^{|E(G)|} \prod_{\{x,y\} \in E(G)} \connf(x-y), \qquad  \weight^\pm (H) := (-1)^{|E^- (H)|} \prod_{\{x,y\} \in E(H)} \connf(x-y). \]

\subsection{The random connection model}

The RCM $\xi$ can be formally constructed as a point process, that is, a random variable taking values in the space of locally finite counting measures $(\mathbf N, \mathcal N)$ on some underlying metric space $\mathbb X$. There are various ways to choose $\mathbb X$. One option is to let $\mathbb X = \Rd \times \mathbb M$ for an appropriate mark space $\mathbb M$ (see~\cite{MeeRoy96}), another way can be found in~\cite{HeyHofLasMat19, LasZie17}. In any way, one can reconstruct from $\xi$ the point process $\eta$ on $\Rd$ which makes up the vertex set of $\xi$. We treat $\eta$ both as a counting measure as well as a set, giving meaning to statements of the form $x \in \eta$.

If $e=\{x,y\}$ is an edge, then we write $\connf(e)= \connf(x-y)$. For a bounded set $\Lambda\subset \Rd$, we write $\eta_\Lambda=\eta\cap\Lambda$ and let $\xi_\Lambda$ denote the RCM restricted to $\Lambda$, that is $\xi[\eta_\Lambda]$. The two-point function restricted to $\Lambda$ is defined as
$\tlam^\Lambda(x,y) = \pla\big(\conn{x}{y}{\xi_\Lambda^{x,y}}\big)$ for $x,y \in \Lambda$ and zero otherwise.

For $V \subset W$, there is a natural way to couple the models $\xi^V$ and $\xi^W$, which is by deleting from $\xi^W$ all points in $W \setminus V$ along with their incident edges. We implicitly assume throughout this paper that this coupling for different sets of added points is used.

\paragraph{The Mecke equation.} Since it is used repeatedly throughout this paper, we state the Mecke equation, a standard tool in point process theory, in its version for the RCM (see~\cite{LasZie17}). For $m \in \N$ and a measurable function $f\colon \mathbf N \times \R^{dm} \to \R_{\geq 0}$, the Mecke equation states that
	\eqq{ \E_\lambda \bigg[ \sum_{\vec x_{[m]} \in \eta^{(m)}} f(\xi, \vec{x}_{[m]})\bigg] = \lambda^m \int
				\E_\lambda\Big[ f\left(\xi^{x_1, \ldots, x_m}, \vec x_{[m]}\right)\Big] \dd \vec{x}_{[m]},  \label{eq:prelim:mecke_m} }
where $\eta^{(m)}=\{\vec x_{[m]} \in \eta^m: x_i \neq x_j \text{ for } i \neq j\}$ are the pairwise distinct tuples.

\paragraph{Re-scaling.} It is a standard trick in continuum percolation to re-scale space in order to normalize a quantity of interest, which is $\int \connf(x) \dd x$ in our case. We point to~\cite[Section 2.2]{MeeRoy96}. As a consequence, we may without loss of generality assume that $\int \connf(x) \dd x=1$.

\paragraph{The BK inequality.} We say that $A\in\mathcal{N}$ {\em lives on} $\Lambda$ if $\mathds 1_{A}(\mu) = \mathds 1_{A}(\mu_\Lambda)$ for every $\mu \in \mathbf N$. We call an event $A \in \mathcal N$ \emph{increasing} if $\mu \in A$ implies $\nu\in A$ for each $\nu\in\mathbf{N}$ with $\mu\subseteq\nu$. Let $\mathcal{R}$ denote the ring of all finite unions of  half-open rectangles with rational coordinates. For two increasing events $A,B\in\mathcal{N}$ we define
	\eqq{ A \circ B := \{ \mu\in\mathbf{N}: \exists K,L \in \mathcal{R} \text{ s.t.~} K \cap L = \varnothing \text{ and } \mu_K \in A,\ \mu_L\in B \}.
\label{def:odot}	} 
Informally, this is the event that $A$ and $B$ take place in spatially disjoint regions. It is proved in~\cite[Theorem~2.1]{HeyHofLasMat19} that for two increasing events $A$ and $B$ living on $\Lambda$, we have 
\[
\pla(A\circ B) \leq \pla(A) \pla(B).\]

\paragraph{The RCM on a fixed vertex set.} Given some (finite) set $X \subset \Rd$ and a function $\connf\colon\Rd\to [0,1]$, we will often have to deal with the following random graph: Its vertex set is $X$, and two vertices $x,y\in X$ are adjacent with probability $\connf(x-y)$, independently of other pairs of vertices. This is simply the RCM conditioned to have the vertex set $X$. To highlight the difference to $\xi$, which depends on the PPP $\eta$, we denote this random graph by $\rg(X)$. If $Y \subset X$, then we write $\rg(Y)$ for $\rg(X)[Y]$. Since there is no dependence on $\lambda$, we write $\p$ for the probability measure of the RCM with fixed vertex set.

\section{Fixing ideas: The two-point function in finite volume} \label{sec:twopoint_fct}

We use this section to put the definitions of Section~\ref{sec:defs} into action and to derive a power-series expansion for $\tlam$ in finite volume. We start by motivating the introduction of ($\pm$)-graphs by linking them to the RCM $\rg$.

\begin{observation}[Connection between $(\pm)$-graphs and probabilities] \label{obs:pm_graphs_probabilities}
Let $X \subset \Rd$ be finite. Let $\mathfrak P \subseteq \mathcal G(X)$ be a graph property. Then
	\[ \sum_{\substack{ G \in \mathcal G^\pm(X): \\ (V(G), E^+(G)) \in \mathfrak P}} \textup{\weight}^\pm(G) 
						= \p \big( \rg(X) \in \mathfrak P  \big). \]
\end{observation}
\begin{proof}
Note that
	\[ \p \big( \rg(X) \in \mathfrak P  \big) = \sum_{\substack{G \in \mathcal G(X): \\G \in\mathfrak P }} \prod_{e \in E(G)} \connf(e) 
					\prod_{e \in \binom{X}{2} \setminus E(G)} (1-\connf(e)).  \]
Expanding the factor $\prod_{e \in \binom{X}{2} \setminus E(G)} (1-\connf(e))$ into a sum proves the claim.
\end{proof}

Note that the weight of a ($\pm$)-graph may also be calculated by taking the product over all its edges, with factors $\connf(\cdot)$ and $-\connf(\cdot)$ for edges in $E^+$ and $E^-$, respectively. Observation~\ref{obs:pm_graphs_probabilities} motivates that the edges in $E^+$ correspond to the edges in the random graph $\rg$. 

Next we prove a power-series expansion for $\tlam$ in terms of the intensity $\lambda$. The expansion~\eqref{eq:tlam_cluster_exp} was already given by Coniglio, De Angelis and Forlani~\cite[Eq.~(12)]{ConDeAngFor77}, who work in the more general context of Gibbs point processes but do not prove convergence. The proposition enters the proof of Proposition~\ref{thm:OZE_finite}. 

Notice that the coefficients of power series expansions like \eqref{eq:tlam_cluster_exp} are given by integrals with respect to the Lebesgue measure and it is sufficient that the integrands are defined up to Lebesgue null sets for those integrals to be well-defined. Since vectors $\vec x_{[3, n+2]}\in\R^{dn}$ with less than $n$ distinct entries constitute a Lebesgue null set, we can assume that for $x_1\neq x_2$ only graphs with vertex sets of cardinality $n+2$ contribute to the $n$-th coefficient in \eqref{eq:tlam_cluster_exp}. The same considerations apply to all graphical expansions appearing from here on, including our main definition \eqref{eq:def:dcf}.

\begin{prop}[Graphical expansion for the two-point function] \label{thm:tlam_finite_cluster_expansion}
Consider the RCM restricted to a bounded measurable set $\Lambda \subset \Rd$, and let $x_1,x_2 \in \Lambda$. Then
	\eqq{ \tlam^\Lambda(x_1,x_2) = \sum_{n \geq 0} \frac{\lambda^n}{n!} \int_{\Lambda^n} \sum_{\substack{G \in {\mathcal C}^\pm(\vec x_{[n+2]}): 
				\\ x_1 \overset{+}{\longleftrightarrow} x_2}} \textup{\weight}^\pm(G) \dd \vec x_{[3, n+2]} \label{eq:tlam_cluster_exp}}
with 
\[
	\sum_{n \geq 0} \frac{\lambda^n}{n!} 
	\int_{\Lambda^n} \Bigl| \sum_{\substack{G \in {\mathcal C}^\pm(\vec x_{[n+2]}): 
				\\ x_1 \overset{+}{\longleftrightarrow} x_2}} \textup{\weight}^\pm(G)\Bigr| \dd \vec x_{[3, n+2]} \leq \exp\Big\{2\lambda + \lambda |\Lambda|\mathrm e^\lambda\Big\} <\infty.
\]				
\end{prop}

Note that Proposition~\ref{thm:tlam_finite_cluster_expansion} is valid for all intensities $\lambda\geq 0$. This situation is completely different from familiar cluster expansions~\cite{brydges1986leshouches}, where the radius of convergence of relevant expansions is finite in finite volume as well.

The expansion~\eqref{eq:tlam_cluster_exp} amounts to the physicists' expansion in powers of the activity. The expansion in powers of the density instead involves sums over a smaller class of graphs. For Poisson point processes, activity and density are the same and the two expansions must coincide. In our context, we point out that the sum over graphs in~\eqref{eq:tlam_cluster_exp} can be reduced to the sum over the subset of graphs in $\mathcal C^\pm$ that contain a $(+)$-path from $x_1$ to $x_2$ and that have no \emph{articulation points} (with respect to~$x_1,x_2$). To define articulation points, recall that a cut vertex leaves a connected graph disconnected upon its deletion. Now, an articulation point is a cut vertex that is not pivotal for the $x_1$-$x_2$-connection. It is not difficult to see that for fixed points $x_{[n+2]}$, the graphs with articulation points in the sum over graphs $G$ in~\eqref{eq:tlam_cluster_exp} exactly cancel out. This cancellation happens at fixed $n$ and does not require any re-summations between graphs with different numbers of vertices.

The proof of Proposition~\ref{thm:tlam_finite_cluster_expansion} builds on yet another equivalent representation: In Eq.~\eqref{eq:tlam_cluster_exp} we can discard those graphs $G$ for which $G^+$ is not connected and those for which not every $(-)$-edge has at least one endpoint in $V(G^+)$, see Eq.~\eqref{eq:rcmfin:tlam_exp_one_sum} below for a precise statement. 
To the best of our knowledge, Eq.~\eqref{eq:rcmfin:tlam_exp_one_sum} is new.

\begin{proof}[Proof of Proposition~\ref{thm:tlam_finite_cluster_expansion}]
We write $\tlam = \tlam^\Lambda$ and $\eta=\eta_\Lambda$. Given $x_1,x_2 \in\Lambda$, we can partition
	\algn{ \tlam(x_1,x_2) &= \sum_{n\geq 0} \pla\big(\conn{x_1}{x_2}{\xi_\Lambda^{x_1,x_2}}, |\C(x_1, \xi_\Lambda^{x_1, x_2})|=n+2\big) \label{eq:tlam_inital_exp} \\
			&= \sum_{n \geq 0} \frac{\lambda^n}{n!} \int_{\Lambda^n} \p\big( \rg(\vec x_{[n+2]})  \in \mathcal C(\vec x_{[n+2]})\big)
					\exp \Big\{ -\lambda \int_\Lambda \big( 1-\prod_{i=1}^{n+2} (1-\connf(x_i-y)) \big) \dd y \Big\} \dd \vec x_{[3,n+2]} \notag. }
The second identity can be found, for example, in~\cite[Proposition 3.1]{LasZie17}. Set
	\[
	f(\vec x_{[n+2]},\vec y_{[m]}) = \p\big( \rg(\vec x_{[n+2]})  \in \mathcal C(\vec x_{[n+2]})\big) \prod_{j=1}^m \Bigl( \prod_{i=1}^{n+2}( 1- \connf(x_i-y_j))-1 \Bigr).
	\]
Expanding the exponential in~\eqref{eq:tlam_inital_exp}, we find 
	\eqq{ \label{eq:tlam-exp-2} 
	\tlam(x_1,x_2) = \sum_{n,m\geq 0}\frac{\lambda^{n+m}}{m! n!} \int_{\Lambda^n}\int_{\Lambda^m} f(\vec x_{[n+2]},\vec y_{[m]}) \dd \vec y_{[m]} \dd \vec x_{[3,n+2]}. }
with 
	\algn{ &\sum_{n,m\geq 0}\frac{\lambda^{n+m}}{m! n!} \int_{\Lambda^n}\int_{\Lambda^m} \bigl| f(\vec x_{[n+2]},\vec y_{[m]})\bigr| \dd \vec y_{[m]} \dd \vec x_{[3,n+2]} \notag \\
		& \qquad = \sum_{n \geq 0} \frac{\lambda^n}{n!} \int_{\Lambda^n} \p\big( \rg(\vec x_{[n+2]})  \in \mathcal C(\vec x_{[n+2]})\big)
			\exp \Big\{ \lambda \int_\Lambda \big( 1-\prod_{i=1}^{n+2} (1-\connf(x_i-y)) \big) \dd y \Big\} \dd \vec x_{[3,n+2]} \notag  \\
		&\qquad \leq \sum_{n \geq 0} \frac{\lambda^n}{n!} \int_{\Lambda^n} \e^{\lambda( n+2)} \dd \vec x_{[3,n+2]} \notag  \\
		&\qquad = \exp\big\{2\lambda + \lambda|\Lambda|\e^\lambda \big\}<\infty. \label{eq:tbb}}
In the third line, we have used the inequality 
	\eqq {\label{ineq:combbound}
	\int_\Lambda \big( 1-\prod_{i=1}^{n+2} (1-\connf(x_i-y)) \big) \dd y \leq \int_\Lambda
	\sum_{i=1}^{n+2} \connf(x_i- y) \dd y \leq n+2,
	}
which can be shown as follows: Let $n\in \mathbb N$ and let $0\leq a_1,\ldots, a_n\leq 1$. Notice that the identity $ 1-\prod_{i=1}^{n} (1-a_i)= (1-a_n)\big(1-\prod_{i=1}^{n-1} (1-a_i)\big)+a_n$ and the estimate $(1-a_n)\leq 1$ hold for all $n\in\mathbb{N}$. The inequality between the integrands in $\eqref{ineq:combbound}$ now follows by induction with the choice $a_i=\varphi(x_i-y)$. The re-scaling introduced in Section 2.3 ensures that  $\int_\Lambda\varphi(x_i-y)\dd y\leq 1$, $i\in[n+2]$, yielding the second inequality. \\

Next we turn to a combinatorial representation of $f$ as a sum over $(\pm)$-graphs. Recall that $\mathcal C^+$ denotes sets of ($\pm$)-graphs that are ($+$)-connected. The definition of $f$ and Observation~\ref{obs:pm_graphs_probabilities} yield 
	\[
	f(\vec x_{[n+2]},\vec y_{[m]}) = \Bigl(\sum_{G \in \mathcal C^{+}(\vec x_{[n+2]})} \weight^\pm(G)\Bigr) \Bigl(
					\sum_{\substack{H \in \mathcal G(\vec x_{[n+2]}, \vec y_{[m]}):\\ y_i \sim \vec x_{[n+2]} \; \forall i \in [m] }} 
					\weight(H)\Bigr)
	 = \sum_{G \oplus H} \weight^\pm(G\oplus H) 
	\]
where the last sum is over all ($\pm$)-graphs $G' = G\oplus H$ in $\mathcal C^\pm(\vec x_{[n+2]} \cup \vec y_{[m]})$ such that first, there are no edges between points of $\vec y$, secondly, $(G \oplus H)^+$ is connected, and thirdly, the vertices of $(G \oplus H)^+$ are precisely $\vec x_{[n+2]}$. 

We re-arrange the double sum~\eqref{eq:tlam-exp-2} over $m,n$ into one sum, indexed by the value of $m+n$, and obtain 
	\algn{\tlam(x_1,x_2) &= \sum_{n \geq 0} \frac{\lambda^n}{n!} \int_{\Lambda^n} 
					\sum_{\substack{G \in \mathcal C^\pm(\vec x_{[n+2]}): \\ \{x_1, x_2\} \subseteq V(G^+), G^+ \text{ connected}, \\
					E(G^{|\pm|}[V \setminus V^+]) = \varnothing}} \weight^\pm(G) \dd \vec x_{[3, n+2]} \label{eq:rcmfin:tlam_exp_one_sum}\\
		&= \sum_{n \geq 0} \frac{\lambda^n}{n!} \int_{\Lambda^n} \sum_{\substack{G \in \mathcal C^\pm(\vec x_{[n+2]}): \\ x_1 \overset{+}{\longleftrightarrow} x_2}} 
					\weight^\pm(G) \dd \vec x_{[3, n+2]}. \label{eq:rcmfin:tlam:final_exp}}
In the second identity, we added some graphs to the sum, namely those in which $G^+$ is not connected or where there exist edges between vertices of $V \setminus V^+$. 

We claim that the weight of these added graphs sums up to zero. To see this, first identify $[n+2]$ with the vertices $\vec x_{[n+2]}$ and fix a graph $G \in \mathcal C([n+2])$. Now, let $C \subseteq [n+2]$ with $\{1,2\} \subseteq C$  and consider the set $\mathcal G_G(C)$ of all $(\pm)$-connected graphs $G^\pm$ on $[n+2]$ so that $G^{|\pm|} = G$ and $C$ is the vertex set of the $(+)$-component of $1$ in $G^\pm$. If there is at least one edge $e$ in $G$ that has both endpoints outside of $C$, we partition $\mathcal G_G(C)$ into those graphs where $e$ is in $E^+$ and those where $e$ is in $E^-$. This induces a pairing between the graphs of $\mathcal G_G(C)$, and they cancel out. What remains are precisely the graphs in~\eqref{eq:rcmfin:tlam_exp_one_sum}.
\end{proof}

\section{The direct-connectedness function} \label{sec:dcfshell}

\subsection{Motivation and rough outline} \label{sec:motivation}

The expansion of the direct-connectedness function in powers of the activity given by~\cite{ConDeAngFor77}, without proofs and convergence bounds, is 
\eqq{ \dcf^\Lambda(x_1, x_2) = \sum_{n \geq 0} \frac{\lambda^n}{n!} \int_{\Lambda^n} 
					\sum_{\substack{G \in \mathcal D_{x_1,x_2}^\pm(\vec x_{[n+2]}): \\ x_1 \overset{+}{\longleftrightarrow} x_2}} 
					\weight^\pm(G) \dd \vec x_{[3, n+2]}. \label{eq:dcf_rewriting}}
It is obtained from the expansion of the pair-connectedness function in Proposition~\ref{thm:tlam_finite_cluster_expansion} by discarding graphs that have pivotal points (i.e., graphs $G$ where $\textsf{Piv}^\pm(G)$ is nonempty). Before we pass to the thermodynamic limit, we perform a resummation and find another representation of $\dcf^\Lambda$ which has the conjectured advantage of increasing the domain of convergence.

Let $G=(V, E^+, E^-) \in \mathcal C^\pm(\vec x_{[n+2]})$ be a $(\pm)$-graph appearing in the expansion~\eqref{eq:rcmfin:tlam_exp_one_sum}. Thus $V= \{ x_i: 1\leq i \leq n+2\}$, the graph $G^+$ is connected, $x_1$ and $x_2$ belong to $V^+=V(G^+)$, every vertex $y\in V(G) \setminus V(G^+)$ is linked by at least one $(-)$-edge to $V^+$, and there are no edges between two vertices in $ V \setminus V^+$. We impose the additional constraint that $G^{|\pm|} = (\vec x_{[n+2]}, E^+\cup E^-)$ has no pivotal points for paths from $x_1$ to $x_2$. 

Since $x_1, x_2$ are connected by a path of $(+)$-edges, $G$ admits a $(+)$-pivot decomposition $\vec W= (u_0, V_0, \ldots, u_k, V_k, u_{k+1})$ (with $u_0=x_1$ and $u_{k+1}=x_2$), where $k \in \N_0$ is the number of pivotal points in $\textsf{Piv}^+(x_1, x_2;G)$. Then, $G$ decomposes into a \emph{core} graph $G_{\text{core}} = (V(G^+), E^+, E_{\text{core}}^-)$ with $E_{\text{core}}^-$ the set of $(-)$-edges of $G$ with both endpoints in $V_i \cup \{u_i, u_{i+1}\}$ for some $i \in [k]_0$, and a \emph{shell} graph $H = (V, \varnothing, E^-\setminus E^-_{\text{core}})$. By our choice of $E^-_{\text{core}}$, we have $\textsf{PD}^\pm(G_{\text{core}}) = \textsf{PD}^+(G_{\text{core}})= \vec W$. Clearly
	\[ 	\weight^\pm (G) = \weight^\pm (G_{\text{core}}) \weight^\pm (H). 	\]
In the right-hand side of~\eqref{eq:dcf_rewriting}, we restrict to graphs that also appear in~\eqref{eq:rcmfin:tlam_exp_one_sum} and rewrite the resulting sum as a double sum over core graphs and shell graphs. This gives rise to the series 
	\[\sum_{r=0}^\infty \frac{\lambda^r}{r!} \int_{\Lambda^r} \sum_{\vec W}\sum_{G_{\text{core}}} \weight^\pm (G_{\text{core}})\Bigl( \sum_{m=0}^\infty \frac{\lambda^m}{m!} 
				\int_{\Lambda^m} \sum_H \weight^\pm (H) \dd \vec y_{[m]} \Bigr) \dd \vec x_{[3,r+2]}.  \] 
The outer sum is over potential pivot decompositions $\vec W$ of core vertices $\vec x_{[r+2]}$, the second sum over $(\pm)$-graphs $G_{\text{core}}= (\vec x_{[r+2]}, E^+,E_{\text{core}}^-)$ that are $(+)$-connected and for which $\vec W$ is both the $(\pm)$-pivot decomposition and the $(+)$-pivot decomposition (in other words, the simple graph $(\vec x_{[r+2]}, E^+)$ is connected and $\textsf{PD}^\pm(x_1,x_2,G) =\textsf{PD}^+(x_1,x_2,G)=\vec W$). The inner sum is over $(\pm)$-graphs $H = (V(H),\varnothing, E^-(H))$ with vertex set $\vec x_{[r+2]}\cup \vec y_{[m]}$ and $(-)$-edges $\{y_i,x_j\}$ such that every vertex $y_i$ is linked to at least one vertex $x_j$, under the additional constraint that $(\vec x_{[r+2]} \cup \vec y_{[m]}, E^+, E_{\text{core}}^-\cup E^-(H))$ has no $(\pm)$-pivotal points for paths from $x_1$ to $x_2$. Let us denote the series associated to such graphs $H$ by $h_\lambda^\Lambda(G_{\text{core}})$,
	\eqq{ h_\lambda^\Lambda(G_{\text{core}}) = \sum_{m=0}^\infty \frac{\lambda^m}{m!} \int_{\Lambda^m} \sum_H \weight^\pm (H) \dd \vec y_{[m]}. \label{eq:shell-prelim}}
The right-hand side of \eqref{eq:shell-prelim} depends on $G_{\text{core}}$ through the pivot decomposition $\vec W$ only. We obtain the representation 
	\eqq{g_\lambda^\Lambda(x_1,x_2) = \sum_{r=1}^\infty \frac{\lambda^r}{r!}\int_{\Lambda^r} \sum_{\vec W} \sum_{G_{\text{core}}} \weight^\pm (G_{\text{core}}) 
			h_\lambda^\Lambda(G_{\text{core}}) \dd \vec x_{[3,r+3]}. \label{eq:gprelim} }
This expression, written in a slightly different form (see Definition~\ref{def:dcfshell:shell_dcf}), forms the starting point of this section. The main results of this section are the following. 
\begin{enumerate} 
\item Let $G_{\text{core}}$ be a $(\pm)$-graph as above. Then the corresponding power series $h_\lambda^\Lambda(G_{\text{core}})$ is absolutely convergent for  all intensities $\lambda \geq 0$ (Proposition~\ref{thm:dcfshell:bounds_on_shell_fct}). In addition, $h_\lambda^\Lambda(G_{\text{core}})$ can be expressed in terms of probabilities involving the random connection model on the fixed vertex set $V(G_{\text{core}})$ and of Poisson processes in $\Lambda$. This alternative expression is used to show that the (pointwise) limit 
	\[	h_\lambda(G_{\text{core}}) = \lim_{\Lambda\nearrow \R^d} h_\lambda^\Lambda(G_{\text{core}})	\] 
exists for all $\lambda>0$ (Lemma~\ref{lem:hlimit}). 
\item Then we show in Theorem~\ref{thm:dcf_inf:limit_theorem} that 
	\[ 	\sum_{r=0}^\infty \frac{\lambda^r}{r!}\int_{(\R^d)^r} \sum_{\vec W} \sum_{G_{\text{core}}} \weight^\pm (G_{\text{core}}) \Bigl|h_\lambda(G_{\text{core}})\Bigr| \dd \vec x_{[3,r+3]}< \infty.  \] 	
This allows us to define 
	\[	g_\lambda(x_1,x_2):= \sum_{r=0}^\infty \frac{\lambda^r}{r!}\int_{(\R^d)^r} \sum_{G_{\text{core}}} \weight^\pm (G_{\text{core}}) h_\lambda(G_{\text{core}}) \dd \vec x_{[3,r+3]} 	\] 
and to pass to the limit in~\eqref{eq:gprelim} showing that
	\[	\lim_{\Lambda\nearrow \R^d} g_\lambda^\Lambda(x_1,x_2) = g_\lambda(x_1,x_2)\] 
as part of Theorem~\ref{thm:dcf_inf:limit_theorem}. 
\end{enumerate} 

\subsection{Definition} 

Here we introduce the precise definitions of core graphs and shell graphs as well as of the functions $h_\lambda^\Lambda$ and $g_\lambda^\Lambda$. We follow the ideas outlined in the previous section but make two small changes. First, shell graphs $H$ are defined not as $(\pm)$-graphs with minus edges only but right away as standard graphs. Second, a close look reveals that the shell function $h_\lambda^\Lambda(G_{\text{core}})$ defined in~\eqref{eq:shell-prelim} depends on the core graph only via $\vec W$; accordingly we view $h_\lambda^\Lambda$ as a function of a sequence of sets.
In addition we drop the index from the core graph; thus the graph $G$ in Definition~\ref{def:graphs:shell_core} below corresponds to $G_{\text{core}}$ in the previous section, see Figure~\ref{fig:New_shell-core_v3}.

\begin{definition}[Core graphs and shell graphs] \label{def:graphs:shell_core}  \
\begin{enumerate}
\item Let $x_1,x_2\in\mathbb R^d$ and let $\{x_1,x_2\}\subset W\subset\Rd$ be a finite set of vertices. We call a graph $G\in\mathcal C^+(W)$ with $\textsf{PD}^\pm(x_1,x_2,G) =\textsf{PD}^+(x_1,x_2,G) = \vec W$ a \emph{core} graph with pivot decomposition $\vec W$ and denote the set of such graphs by $\mathcal G^{\vec W}_{\text{core}}$.
\item Let $G\in\mathcal C^+(W)$ be a core graph with pivot decomposition $\vec W=(u_0, V_0, \ldots, V_k, u_{k+1})$, $k\in\mathbb{N}_0$, where we set $u_0:=x_1$ and $u_{k+1}:=x_2$. Moreover, let $\overline V_i := V_i \cup \{u_i, u_{i+1}\}$ and let $Y$ be a finite subset of $\mathbb{R}^d$. A \emph{shell} graph on $W\cup Y$ associated to $\vec W$ is a $(k+1)$-partite graph $H\in \mathcal G(\overline V_1, \ldots, \overline V_k, Y)$ such that $G \oplus H \in \mathcal D_{x_1,x_2}^\pm(W \cup Y)$. We call the vertices $Y\subset V(H)$ satellite vertices and write $S(H)=Y$. Notice that the set of all shell graphs on $W\cup Y$ associated to $\vec W$ does not depend on the choice of the core graph $G$. We denote it by $\mathcal G^{Y, \vec W}_{\text{shell}}$. 
\end{enumerate}
\end{definition}

\begin{figure}
	\centering
\includegraphics[scale=0.5]{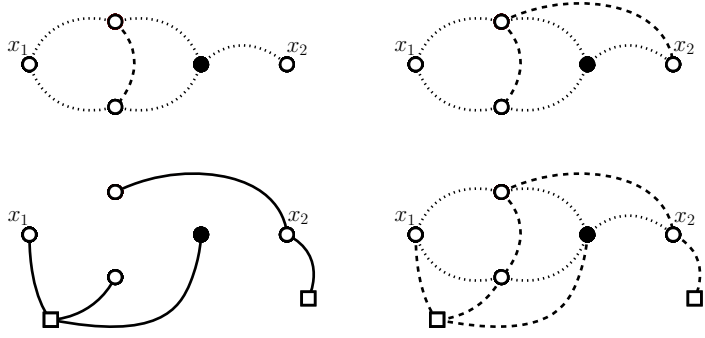}
\caption{In the first line, we see an example of two ($\pm$)-graphs; the ($+$)-edges are depicted by dotted and the ($-$)-edges by dashed lines. Notice that both graphs are ($+$)-connected. However, the solid black vertex---which is ($+$)-pivotal for the $x_1$-$x_2$-connection in both graphs---is ($\pm$)-pivotal for the $x_1$-$x_2$-connection in the graph on the left but not in the graph on the right. Hence, the graph on the left is a core graph according to Definition \ref{def:graphs:shell_core} but the graph on the right is not.  In the second line, the simple graph on the left is a shell graph for the core graph above, since the ($\pm$)-graph given by their sum (depicted on the right) is ($\pm$)-doubly connected, in particular there are no ($\pm$)-pivotal points for the $x_1$-$x_2$-connection.\label{fig:New_shell-core_v3}}
\end{figure}
We define $h_\lambda^\Lambda$ and $g_\lambda^\Lambda$ by expansions similar to~\eqref{eq:shell-prelim} and~\eqref{eq:gprelim} and postpone the proof of convergence to Proposition~\ref{thm:dcfshell:bounds_on_shell_fct} and Theorem~\ref{thm:dcf_inf:limit_theorem}.
By some abuse of language, we refer to the series~\eqref{eq:def:dcf} as the direct-connectedness function, and we use the same letter $g_\lambda$ as in~\eqref{eq:intro:oze_2d}. This is justified a posteriori by the proof of Theorem~\ref{thm:main_thm}, where we show that the series is indeed the expansion for the direct-connectedness function $\dcf$ defined as the unique solution of the OZE~\eqref{eq:intro:oze_2d}.

\begin{definition}[Shell functions and direct-connectedness function] \label{def:dcfshell:shell_dcf} \
\begin{enumerate}

\item Let $W\subset \mathbb R^d$ be finite and $\vec W$ be given as in Definition \ref{def:graphs:shell_core}. For $m\in\mathbb N_0$, define the $m$-\emph{shell function} $h^{(m)}$ by
	\eqq{ h^{(m)}(\vec W, Y) :=  \sum_{\substack{H \in \mathcal G^{Y, \vec W}_{\text{shell}}}}  \normalfont{\weight}(H) ,\quad Y=\{y_1,\ldots,y_m\}\subset \R^d, \label{eq:def:mshell_fct}}
and the \emph{shell function} $h^\Lambda_\lambda$ in finite volume $\Lambda\subset \R^d$ by 
	\eqq {h^\Lambda_\lambda(\vec W)  := \sum_{m \geq 0} \frac{\lambda^m}{m!} \int_{\Lambda^m}  h^{(m)}(\vec W, \vec y_{[m]}) \dd \vec y_{[m]}.\label{eq:def:shell_fct}}
\item Let $\lambda < \lambda_\ast$. We define the \emph{direct-connectedness function} as $\dcf\colon \Rd\times\Rd \to \R$, 
	\eqq{ \dcf^{\Lambda}(x_1,x_2) := \sum_{r \geq 0} \frac{\lambda^r}{r!} \int_{\Lambda^{r}} \sum_{\vec W} 
				\Big(\sum_{G \in \mathcal G^{\vec W}_{\text{core}}}
				\weight^\pm(G) \Big) h_\lambda^\Lambda(\vec W)  \dd \vec x_{[3,r+2]}, \label{eq:def:dcf}}
where  $W:=\{x_1,\ldots,x_{r+2}\}$ and we sum over decompositions $\vec W$ of $W$ given as in Definition \ref{def:graphs:shell_core}. In the pathological case $x_1=x_2$, $\eqref{eq:def:dcf}$ is to be read as $\dcf^{\Lambda}(x_1,x_2) := 1$. Let $\dcf^{\Lambda}\colon\Rd \to \R$ be defined by $\dcf^{\Lambda}(x) = \dcf^{\Lambda}(\orig, x)$. 
\end{enumerate}
\end{definition}

The $0$-shell function $h^{(0)}$ is understood to be given in terms of shell graphs without satellite vertices, i.e.,
	\[ h^{(0)}(\vec W) = \sum_{\substack{H \in \mathcal G^{\varnothing, \vec W}_{\text{shell}}}}  \normalfont{\weight}(H). \] 
Note that due to translation invariance, $\dcf^\Lambda(x_1,x_2) = \dcf^\Lambda(\orig, x_2-x_1) = \dcf^{\Lambda}(x_2-x_1) $.

\subsection{Analysis of the shell functions: Laces} \label{sec:laces}

If we take a look at the graphs that are summed over in the shell function, we note that the associated minimal structures have a form which is very reminiscent of graphs that are known as \emph{laces} and famously appeared in the analysis of, for example, self-avoiding walk~\cite{BrySpe85, Sla06}. They are also the namesake of the lace-expansion technique.

Proposition~\ref{thm:dcfshell:bounds_on_shell_fct} is the central result of this section. It allows to bound the shell function by the probability that the points in a PPP $\eta$ are not connected to the core vertices $W$. Moreover, we introduce laces and partition the shell graphs with respect to them. For every lace, we obtain a precise expression of its contribution to the shell function.

To prove Proposition~\ref{thm:dcfshell:bounds_on_shell_fct}, we will need quite a few definitions (see Definitions~\ref{def:laces:skeletons},~\ref{def:laces},~\ref{def:laces:span}) and some intermediate results thereon.

\begin{prop}[Bounds on the shell functions] \label{thm:dcfshell:bounds_on_shell_fct}
Let $\lambda\geq 0$ and let $\Lambda\subset \Rd$ be bounded. Let $u_0, \ldots, u_{k+1} \in \Lambda$ for $k\in\N_0$, let $V_0, \ldots, V_k \subset \Lambda$ be finite sets, and set $\vec W=(u_0, V_0, \ldots, V_k, u_{k+1})$. Then
	\eqq{|h_\lambda^\Lambda(\vec W)| \leq \pla(\nconn{\eta_\Lambda}{W}{\xi^{W}}). \label{eq:dcfshell:main_bound}}
Moreover,
	\eqq{ \sum_{m \geq 0}\frac{\lambda^m}{m!} \int_{\Lambda^m} \big| h^{(m)} (\vec W, \vec y_{[m]}) \big| \dd\vec y_{[m]} 
			\leq \frac{1}{\sqrt{5}} \textup{\e}^{3 \lambda|W|} \big( 3+\sqrt{5} \big)^{|W|}. \label{eq:dcfshell:m_shell_summability} }
\end{prop}

Proposition~\ref{thm:dcfshell:bounds_on_shell_fct} consists of two parts, and it is~\eqref{eq:dcfshell:m_shell_summability} that guarantees the well-definedness of the shell function $h_\lambda^\Lambda$ of Definition~\ref{def:dcfshell:shell_dcf}. 

Proposition~\ref{thm:dcfshell:bounds_on_shell_fct} is easy to prove for $k=0$, and we mostly focus on $k\geq 1$. Throughout the remainder of this section, we fix a pivot decomposition $\vec W = (u_0, V_0, \ldots, V_k, u_{k+1})$ and recall that $\overline{V}_i = V_i \cup \{u_i,u_{i+1}\}$.

We now work towards a deeper understanding of the shell graphs $H$ summed over in~\eqref{eq:def:mshell_fct}. 
\begin{definition}[Skeletons] \label{def:laces:skeletons}
Let $W\subset \R^d$ and let $\vec W=(u_0,V_0,\ldots, u_{k+1})$ be a pivot decomposition of some core graph on $W$. Further let $Y\subset \R^d$ be finite and let $H$ be a shell graph associated to $\vec W$ with satellite vertices $S(H)=Y$. Then we define the \emph{skeleton} $\hat H$ of $H$ as the following graph: Its vertex set is $V(\hat H) = \{0,\ldots, k+1\}$. A bond $\alpha\beta$ is in $E(\hat H)$ if and only if $\vert\alpha-\beta\vert\geq 2$ and there exist $s\in \{u_\alpha\}\cup V_\alpha$, $t\in V_{\beta-1}\cup\{u_{\beta}\}$ such that	
	\begin{compactitem}
		\item  $st\in E(H)$, or
		\item $sy$, $yt\in E(H)$ for some $y\in S(H)$.
	\end{compactitem}
In the first case we call $\{s,t\}$ a \emph{direct stitch} and in the second case an \emph{indirect stitch}. We call an edge $\alpha\beta$ in $E(\hat H)$ a \emph{bond} to distinguish it from the edge of the underlying graph $H$.
\end{definition}

Thus, the graph $\hat H$ has no nearest-neighbor bonds and $\alpha\beta$ with $\vert\alpha-\beta\vert\geq 2$ is a bond in $E(\hat H)$ if and only if $\{u_\alpha\}\cup V_\alpha$ and $V_{\beta-1}\cup\{u_{\beta}\}$  are connected by a direct or indirect stitch. See Figure~\ref{fig:laces_equivalence_classes} for an illustration. We may now apply the standard vocabulary of lace expansion (for self-avoiding walks) to the graph $\hat H$~\cite[Section 3.3]{Sla06}. 
		
\begin{definition}[Laces] \label{def:laces} \hfill
\begin{compactitem}
\item The graph $\hat H$ with vertex set $\{0,\ldots, k+1\}$ is \emph{irreducible} if $0$ and $k+1$ are endpoints of edges in $E(\hat H)$ and for every $i\in [k]$ there exists $\alpha\beta\in E(\hat H)$ with $\alpha<i< \beta$.
\item The graph $\hat H$ is a \emph{lace} if it is irreducible and, for every bond $\alpha\beta\in E(\hat H)$, removal of the bond destroys the irreducibility.
\item We denote by $\mathcal L_k$ the set of all laces on $\{0, \ldots, k+1\}$. 
	\end{compactitem} 
\end{definition}
In the context of lace expansions, usually the word ``connected'' is used instead of ``irreducible'', but ``connected'' is clearly misleading in our setup; Brydges and Spencer originally called those graphs ``primitive''~\cite{BrySpe85}. We observe that the skeleton graphs $\hat H$ arising from our shell graphs $H$ are precisely the irreducible graphs (and so $G\oplus H$ being $2$-connected corresponds to the skeleton $\hat H$ being irreducible).

	We map irreducible graphs to laces by following a standard procedure~\cite[Section 3.3]{Sla06}, performed backwards. That is, we define bonds $\alpha'_j\beta'_j$ with $\beta'_1>\beta'_2>\cdots$ inductively as follows: we set 
	\[	\beta'_1 :=k+1,\quad \alpha'_1 := \min\{\alpha:\ \alpha\beta'_1\in E(\hat H)\}. \] 
and 
	\[	\alpha'_{j+1} = \min \{\alpha:\ \exists \beta > \alpha'_j\text{ with } \alpha \beta\in E(\hat H)\},	\quad
			\beta'_{j+1} = \max \{\beta:\ \alpha'_{j+1} \beta \in E(\hat H)\}.\]
The procedure terminates when $\alpha'_j = 0$. At the end, we let $\alpha_j\beta_j$ be a relabelling of the bonds $\alpha'_j \beta'_j$ from left to right. 

It is well-known that the algorithm maps irreducible graphs to laces, moreover the set of irreducible graphs that are mapped to a given lace $L$ can be characterized as follows.

\begin{figure}
	\centering
	\includegraphics[scale=0.9]{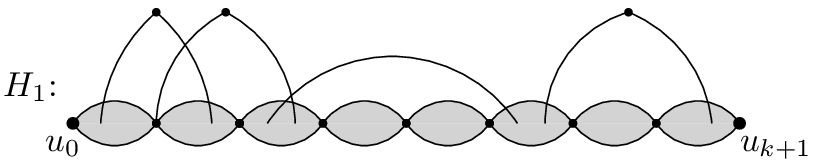} \quad
	\includegraphics[scale=0.9]{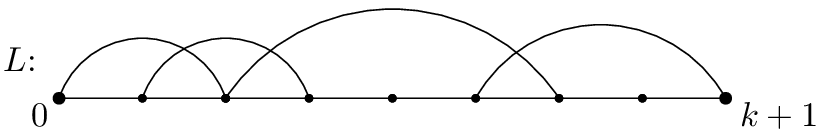} \\ \vspace{0.5cm}
	\includegraphics[scale=0.9]{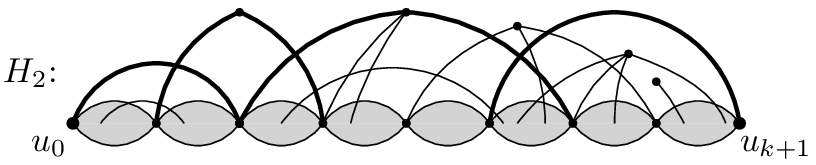} \quad
	\includegraphics[scale=0.9]{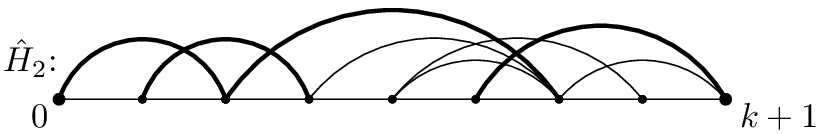}
	\caption{In the first line, we see a schematic shell graph $H_1$. Its skeleton $\hat H_1$ is already a lace, namely $L$. The skeleton of the graph $H_2$ in the second line is not a lace, but $H_2 \in \gen{L}$. The structure of $L$ is indicated in $H_2$ and in $\hat H_2$ by the thicker edges.}
	\label{fig:laces_equivalence_classes}
\end{figure}

\begin{definition}[Compatible bonds and the span of a lace] \ \label{def:laces:span}
\begin{enumerate}
\item Let $L$ be a lace with vertex set $\{0,\ldots,k+1\}$. A bond is \emph{compatible} with a lace $L$ if the algorithm described above maps the graph $(V(L), E(L)\cup \{\alpha\beta\})$ to the lace $L$.
\item Let $W\subset \R^d$ and let $\vec W=(u_0,V_0,\ldots, u_{k+1})$ be a pivot decomposition of some core graph on $W$. Further let $Y\subset \R^d$ be finite and let $H$ be a shell graph associated to $\vec W$ with $S(H)=Y$. Then we say that $H$ belongs to the span of the lace $L$, written $H\in \gen{L}$, if $E(L) \subseteq E(\hat H)$ and every bond $\alpha\beta \in E(\hat H) \setminus E(L)$ is compatible with $L$.
\end{enumerate}

\end{definition} 

In other words, $H$ is in the span of $L$ if the above algorithm maps $\hat H$ to $L$. See Figure~\ref{fig:laces_equivalence_classes}.

Given $\vec W$ and a lace $L$, we define
		\eqq{ h_\lambda^\Lambda(\vec W; L) := \sum_{m \geq 0} \frac{\lambda^m}{m!} \int_{\Lambda^m} \sum_{H \in \gen{L}: \mathcal S(H) = \vec y_{[m]}}
				 	\weight(H) \dd \vec y_{[m]}. \label{eq:def:laces:shell_fixed_lace}}
The series $h_\lambda^\Lambda(\vec W;L)$ converges absolutely for every fixed $\lambda$. This is shown as part of the proof of~\eqref{eq:dcfshell:m_shell_summability} in Proposition~\ref{thm:dcfshell:bounds_on_shell_fct}. Now,
	\[h_\lambda^\Lambda(\vec W) = \sum_{L \in \mathcal L_k} h_\lambda^\Lambda(\vec W; L). \]

The following characterization of compatible bonds will be useful. We recall that the bonds of a lace with $m$ bonds can be labelled as $\alpha_j\beta_j$ with 
\[
	0 = \alpha_1 < \alpha_2 < \beta_1\leq \alpha_3 < \beta_2\leq \cdots \leq \alpha_m < \beta_{m-1}< \beta_m = k+1.
\] 
see~\cite[Eqs.~(3.15) and~(3.16)]{Sla06}.

\begin{lemma} [Characterization of compatible bonds] \label{lem:compatible}
Let $L$ be a lace with vertex set $V(L) = \{0,\ldots, k+1\}$ and bonds $\alpha_j\beta_j$, $j=1,\ldots, m$, labeled from left to right (i.e., $\alpha_j<\alpha_{j+1}$). Then, a bond $\alpha\beta\notin E(L)$ with $\alpha <\beta-1$ is compatible with $L$ if and only if either
	\begin{compactitem}
	\item[(a)] $\alpha_i \leq \alpha <\beta \leq \beta_i$ for  $i\in[m]$ or
	\item[(b)] $\alpha_i < \alpha<\beta \leq \alpha_{i+2}$ for  $i\in[m-1]$ (where we set $\alpha_{m+1} :=k$). 
\end{compactitem}
\end{lemma} 	

\begin{proof} 
Let $\alpha\beta \notin E(L)$ be compatible with $L$; that is, the algorithm below Definition~\ref{def:laces} maps $E(L) \cup \{\alpha\beta\}$ to $E(L)$, which in turn means that $\alpha\beta$ is not selected to be part of the output lace. We show that then, either (a) or (b) is satisfied. Assume the algorithm has already constructed the partial lace up to some $j < m$, producing the bonds $(\alpha_i', \beta_i')_{i=1}^j$ (note that they are in reverse order and make up the last $j$ bonds of the lace). Assume moreover that $\alpha_j'< \beta \leq \alpha_{j-1}'$, that is, $\alpha\beta$ is a potential candidate to be chosen as the next bond of the lace. Since it is \emph{not} chosen, there is $\alpha_{j+1}'\beta_{j+1}'$ with $\beta_{j+1}' \in (\alpha_j',\alpha_{j-1}']$ so that either
\begin{enumerate}
\item[$\bullet$] $\alpha_{j+1}'<\alpha$, or 
\item[$\bullet$] $\alpha_{j+1}' = \alpha$ and $\beta_{j+1}' > \beta$.
\end{enumerate}
Both the second case as well as the first case under the additional assumption $\beta_{j+1}' \geq \beta$ imply that $\alpha\beta$ satisfies (a). Let us thus focus on the case where $\alpha_{j+1}'<\alpha$ and $\beta_{j+1}' < \beta$. Remembering the stage of the algorithm, we have $\beta \leq \alpha_{j-1}'$, implying (b).

Let now $\alpha\beta \notin E(L)$ be a bond that satisfies (a) or (b). We claim that $\alpha\beta$ is compatible with $L$. Let $i$ be the index such that $\alpha_i \beta_i$ satisfies (a) or (b). Note that in the execution of the algorithm below Definition~\ref{def:laces}, $\alpha\beta$ does not appear as a candidate to be added to the constructed lace up until the point where $\alpha_m\beta_m, \alpha_{m-1}\beta_{m-1}, \ldots, \alpha_{i+1}\beta_{i+1}$ are already added to the partial lace. At this stage of the algorithm, if $\alpha\beta$ satisfies (b), then it is not picked because the bond $\alpha_i\beta_i$ has a smaller value of its left endpoint (i.e., $\alpha_i<\alpha$). If $\alpha\beta$ satisfies (a) however, then either also $\alpha_i<\alpha$, or $\alpha_i=\alpha$, but $\alpha_i\beta_i$ has its right endpoint further to the right (i.e., $\beta<\beta_i$, since the two bonds cannot be equal), and so again, $\alpha_i\beta_i$ is picked by the algorithm.
\end{proof}

To prove the second result of Proposition~\ref{thm:dcfshell:bounds_on_shell_fct}, we need the following counting lemma, which may be of independent interest.

\begin{lemma}[On the number of laces] \label{lem:laces:size_L}
Let $f_i$ be the $i$-th Fibonacci number with $f_1 =0, f_2=1$. Then
	\[ |\mathcal L_k| = 1 + \sum_{i=1}^{k} \binom{k}{i} f_i \qquad \text{and, as } k \to\infty, \qquad |\mathcal L_k| \sim \frac{1}{\sqrt{5}} \Big( \frac{3+\sqrt{5}}{2} \Big)^k.\]
\end{lemma}
\begin{figure}
	\centering
 \includegraphics{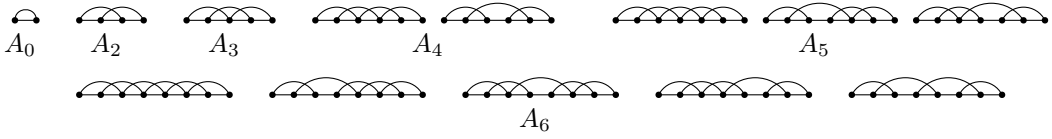}
	\caption{Schematic illustration of $A_i$ from the proof of Lemma~\ref{lem:laces:size_L} for $i=0,2,3,4,5,6$.}
	\label{fig:laces_count}
\end{figure}

\begin{proof}
We first choose $i$ vertices in $\{1, \ldots, k\}$ and then count the laces that use exactly those vertices. To this end, let $A_i$ be the set of laces $L$ with $V(L)=\{0, \ldots, i+1\}$ so that every vertex is the endpoint of at least one stitch. We claim that $|A_i|=f_i$ for $i \geq 1$. Clearly, $|A_0|=1, |A_1|=0, |A_2|=1$. See Figure~\ref{fig:laces_count} for an illustration.

Let $i \geq 3$. We now establish the Fibonacci recursion. First, note that the bond incident to $0$ (the ``first'' bond) must always have $2$ as the second endpoint. Now, depending on whether or not the third bond is incident to $2$, the remaining lace lives on $\{1, 2, \ldots, i+1\}$ or on $\{1, 3, 4, \ldots, i+1\}$, and so $|A_i| = |A_{i-1}| + |A_{i-2}|$.

The asymptotic behavior follows from the fact that $f_n \sim \Phi^n /\sqrt{5}$, where $\Phi = \tfrac 12 (1+\sqrt{5})$ is the golden ratio.
\end{proof}

We can now work towards finding an explicit expression for $h_\lambda^\Lambda(\vec W;L)$ for a fixed lace. The next lemma is in the spirit of Observation~\ref{obs:pm_graphs_probabilities} and will help us find probabilistic factors in the shell function.

\begin{lemma}[Bipartite graphs and probabilities] \label{lem:laces:weight_sums_probabilities}
Let $Y, A, B,C \subset \Rd$ be finite, disjoint sets.
\begin{enumerate}
\item Then \[ \sum_{\substack{H \in \mathcal G(A\cup C,Y): \\ \forall y \in Y: y \sim A}} \textup{\weight} (H) 
				= \prod_{y \in Y} \big( -\p(A \sim y \nsim C) \big) 
				= (-1)^{|Y|} \p(\forall y \in Y: A \sim y\nsim C).\]
\item Moreover,
	\[ \sum_{\substack{H \in \mathcal G(A \cup B \cup C,Y): \\ \forall y \in Y: A \sim y \sim B}} \textup{\weight} (H)
				= \prod_{y\in Y}\p\big(A \sim y \sim B, y \nsim C \big). \]
\item Lastly,
	\[ \sum_{\substack{H \in \mathcal G(A,Y): \\ E(H) \neq \varnothing}} \textup{\weight} (H) = - \p(A \sim Y). \]
\end{enumerate}
\end{lemma}

\begin{proof}
The first part of the statement is rather straightforward. If $Y=\{y\}$, then $\mathcal G(A\cup C, \{y\})$ is the set of star graphs (with center $y$). Observe first that
	\[ \sum_{H \in\mathcal G(A \cup C, \{y\}): y \sim A} \weight(H) = \Big( \sum_{H' \in\mathcal G(A, \{y\}): y \sim A} \weight(H') \Big)
				\Big( \sum_{H'' \in\mathcal G(C, \{y\})} \weight(H'') \Big).\]
The first sum is over all star graphs in $\mathcal G(A,\{y\})$ except the empty one, the second is over all star graphs in $\mathcal G(C,\{y\})$, and so
	\[ \sum_{H \in\mathcal G(A \cup C, \{y\}): y \sim A} \weight(H) = - \big(1- \prod_{x \in A}(1- \connf(y,x))\big) \prod_{x \in C}(1- \connf(y,x)) = -\p(A \sim y \nsim C).\]
It is an easy induction to prove that for general $Y$, the sum factorizes into a product over sums over star graphs. For the second statement, assume again that $Y = \{y\}$ and observe that
	\[\sum_{\substack{H \in\mathcal G(A \cup B \cup C, \{y\}): \\ A \sim y \sim B}} \weight(H) =
		\Big( \sum_{H \in\mathcal G(A \cup C, \{y\}): y \sim A} \weight(H) \Big) \Big( \sum_{H \in\mathcal G(B, \{y\}): y \sim B} \weight(H) \Big) = \p(A \sim y \sim B, y \nsim C),\]
where the last identity is due to independence. The statement easily extends to general $Y$ (again, the sum factorizes).

For the third statement, note that we sum over every graph except the empty one.
\end{proof}

Since the explicit expression for $h_\lambda^\Lambda(\vec W;L)$ is a lengthy product of probabilities, we first introduce some notation to represent the factors of this product compactly. Let $A,B$ be two subsets of $[k+1]_0$. We define the set of all possible direct stitches in $H$ leading to bonds $\alpha \beta\in E(\hat H)$ with $\alpha\in A$, $\beta \in B$, as
\[
	 \Upsilon(A,B) := \Big\{ xy \subset W:\ \exists \alpha \in A, \beta\in B \text{ with } \alpha<\beta-1 \text{ and }  x\in \{u_\alpha\}\cup V_\alpha,\ y \in V_{\beta-1}\cup \{u_\beta\} \Big\}
\]
and we write $\Upsilon(A) = \Upsilon(A,A)$. We define
	\[ q_{\alpha,\beta} := \prod_{xy \in \Upsilon([\alpha,\beta)) \cup \Upsilon((\alpha,\beta])}
				(1-\connf(x-y))\]
and, for $\alpha_1<\alpha_2<\alpha_3$,
	\[ q_{\alpha_1, \alpha_2, \alpha_3} := \prod_{xy \in \Upsilon([\alpha_1+1,\alpha_2),[\alpha_2, \alpha_3))} (1-\connf(x-y)).\]
Note that these products encode the sum over all $\weight$-weighted graphs on the set of edges multiplied over. 

To lighten notation, for $0 \leq\alpha\leq \beta\leq k+1$, set
\begin{gather*}
	[u_\alpha]\!] :=\{u_\alpha\}\cup V_\alpha,\qquad [\![u_\beta] := V_{\beta-1} \cup \{u_\beta\},\\
	[u_\alpha,u_\beta]:= \{u_\alpha\} \cup V_\alpha \cup \cdots \cup V_{\beta-1}\cup \{u_\beta\}.
\end{gather*}
We extend this notation further: For $a,b \in \{u_0, \ldots, u_{k+1}\}$, let $(a,b) := [a,b] \setminus \{a, b\}$, let $[a,b) := [a,b] \setminus \{b\}$, and $(a,b]:= [a,b]\setminus\{a\}$. We set $(\!(a,b)\!) := [a,b] \setminus ([a]\!] \cup [\![b]) $ and define sets $(\!(a,b]$ etc.~accordingly. 

Moreover, define
	\[ Q_{\alpha,\beta} = \pla\big(\nexists y \in \eta_\Lambda \text{ s.t. } [u_\alpha]\!] \sim y \sim [\![u_\beta], y \nsim [u_{\alpha+1}, u_{\beta-1}] \big) \]
for $\beta \geq \alpha+2$. We extend this notation by writing
	\[ Q_{A,B} = \prod_{\alpha  \in A} \prod_{\beta \in B} Q_{\alpha,\beta}\]
for sets of pivotal points $A,B$; we abbreviate $Q_{a,[b,c]} = Q_{\{a\},[b,c]}$. 

We are now ready to state Lemma~\ref{lem:laces:shell_fct_fixed_lace}, for which we recall the definition of $h_\lambda^\Lambda(\vec W;L)$ in~\eqref{eq:def:laces:shell_fixed_lace}.

\begin{lemma}[The shell function of a lace] \label{lem:laces:shell_fct_fixed_lace}
Let $\lambda \geq 0$ and let $\Lambda \subset \Rd$ be bounded. Let $W\subset \Rd$ be a core vertex set with pivot decomposition $\vec W=(u_0, V_0, \ldots, u_{k+1})$. Let $L$ be a lace with vertex set $[k+1]_0$ and $m$ bonds $\alpha_i\beta_i$, $i \in [m]$. Then, setting $\alpha_{m+1}=k$,
	\algn{  h_\lambda^\Lambda(\vec W;L) &= \pla(\eta_\Lambda \centernot\longleftrightarrow W) \prod_{i=1}^{m} q_{\alpha_i,\beta_i}  
				\Big[1- Q_{\alpha_i,\beta_i} - \p([u_{\alpha_i}]\!] \sim [\![ u_{\beta_i}]) \Big] \notag\\
		& \qquad\qquad \times \prod_{i=1}^{m-1} q_{\alpha_i,\alpha_{i+1},\alpha_{i+2}} Q_{\alpha_i,(\beta_i,k+1]} Q_{(\alpha_i,\alpha_{i+1}),(\alpha_{i+2},k+1]}.
			\label{eq:dcfshell:hWL_lemma_identity}}
Moreover,
	 \eqq{ \sum_{n \geq 0} \frac{\lambda^n}{n!} \int_{\Lambda^n} \Big| \sum_{H \in \gen{L}: \mathcal S(H) = \vec y_{[n]}} \textup{\weight}(H) \Big| \dd \vec y_{[n]} 
	 		\leq 2^m \textup{\e}^{3\lambda|W|}	\label{eq:dcfshell:hWL_lemma_abs_conv}}
\end{lemma}

\begin{figure}
	\centering
 \includegraphics[scale=0.9]{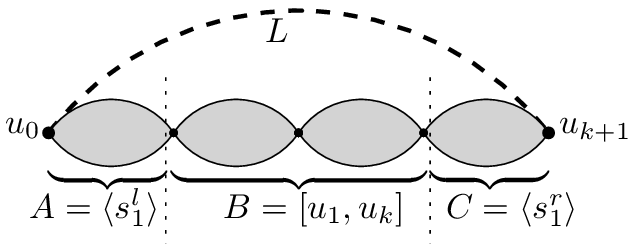} \quad
 \includegraphics[scale=0.9]{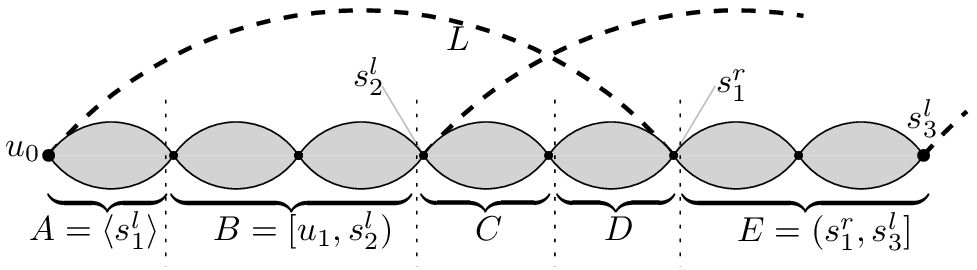}
	\caption{Illustration of the induction proof of Lemma~\ref{lem:laces:shell_fct_fixed_lace}. The lace $L$ is sketched by the dashed lines. The left picture shows the base case $m=1$, where $u_0=s_1$ and $u_{k+1}=t_1$. To the right, the first three stitches of $L$ are (partially) sketched. The sets $C,D$ are defined as $C = [s_2,t_1)\!)$ and $D=[\![t_1]$.}
	\label{fig:h_bound_base}
\end{figure}

\begin{proof} We abbreviate $\eta=\eta_\Lambda, h = h_\lambda^\Lambda$, and prove the statement by induction on $m$.

\underline{Base case.} Let $m=1$, which means that $\alpha_1=0$ and $\beta_1=k+1$. Set $A=[u_0]\!], B=[u_1, u_k]$, and $C=[\![u_{k+1}]$. See Figure~\ref{fig:h_bound_base} for an illustration of $A,B,C$.

Note first that the edge set $\Upsilon([k+1]_0) \setminus E(A,C)$, that is, the possible direct stitches between points of $W$ except the direct ones between $A$ and $C$, do not determine membership of $H$ in $\gen{L}$. Any such edge $xy$ may or may not be present, resulting in a factor $(1-\connf(x-y))$ that can be extracted. In total, this produces the factor $q_{0,k+1}$ and we can restrict to considering graphs $H\in\gen{L}$ that do not possess any such edge. The remaining graphs $H$ only have edges that are incident to $A\cup C\cup \mathcal S(H)$.

We split this set of remaining graphs $H$ into those that have a direct stitch between $A$ and $C$ and those that do not. Among the former, the sum over graphs factorizes into graphs $H' \in \mathcal G(A,C)$ (the direct stitches) and graphs $H'' \in \mathcal G(W, \mathcal S(H))$. With Lemma~\ref{lem:laces:weight_sums_probabilities},
	\algn{h(\vec W;L) &= q_{0,k+1} \bigg[\sum_{n \geq 0} \frac{\lambda^n}{n!} \int_{\Lambda^n}\Big(\sum_{H' \in \mathcal G(A,C): E(H') \neq \varnothing} \weight(H') \Big) 
					\Big(\sum_{\substack{H'' \in \mathcal G(W,\vec y_{[n]}): \\ y_i \sim W \forall i \in[n]}} \weight(H'') \Big) \dd \vec y_{[n]} \notag\\
		& \hspace{2cm} + \sum_{n \geq 0} \frac{\lambda^n}{n!}	\int_{\Lambda^n}
				\sum_{\substack{H \in \mathcal G(W,\vec y_{[n]}): \\ \deg(y_i) \geq 1 \forall i \in[n], \\ \exists i: A \sim y_i \sim C}} \weight(H) \dd \vec y_{[n]} \bigg] \notag \\
		&= q_{0,k+1} \bigg[-\p(A\sim C) \pla(\eta \centernot\longleftrightarrow W) + \sum_{n \geq 0} \frac{\lambda^n}{n!}	\int_{\Lambda^n}
				\sum_{\substack{H \in \mathcal G(W,\vec y_{[n]}): \\ \deg(y_i) \geq 1 \forall i \in[n], \\ \exists i: A \sim y_i \sim C}} \weight(H) \dd \vec y_{[n]} \bigg]
							 \label{eq:dcfshell:hWL_base_split}.}
For now the power series are treated as formal power series, convergence is proven later.
To treat the sum in~\eqref{eq:dcfshell:hWL_base_split}, we define
	\[ \mathcal S_1 := \{y: A \sim y \sim C\}, \quad \mathcal S_2:= \{y: C \nsim y \sim (A \cup B)\}, \quad \text{and } \mathcal S_3:= \{ y: C \sim y \nsim A \}. \]
With these definitions, we can partition $\vec y = \mathcal S(H)= \mathcal S_1 \cup \mathcal S_2 \cup \mathcal S_3$. Moreover, we know that $\mathcal S_1 \neq \varnothing$. Re-summing and then applying Lemma~\ref{lem:laces:weight_sums_probabilities}, the sum over $n$ in~\eqref{eq:dcfshell:hWL_base_split} becomes
	\algn{& \sum_{n_1,n_2,n_3 \geq 0} \frac{\lambda^{n_1+n_2+n_3}}{n_1!n_2!n_3!} \int_{\Lambda^{n_1+n_2+n_3}}
			\sum_{\substack{H \in \gen{L}: \\ \mathcal S_i (H)=\vec y_ {i,[n_i]} \forall i \in [3]}} \weight(H)
			\dd \big(\vec y_{1,[n_1]}, \vec y_{2,[n_2]}, \vec y_{3,[n_3]} \big) \notag\\
		=&  \bigg( \sum_{n \geq 1} \frac{\lambda^n}{n!} \int_{\Lambda^n} \sum_{\substack{H \in \mathcal G(A \cup B \cup C,\vec y_{[n]}): \\ \forall i \in [n]: A \sim y_i \sim C }}
			 	\weight(H) \dd \vec y_{[n]} \bigg)
			 \bigg( \sum_{n \geq 0} \frac{\lambda^n}{n!} \int_{\Lambda^n} \sum_{\substack{H \in \mathcal G(A \cup B,\vec y_{[n]}): \\ \forall i \in [n]: y_i \sim (A \cup B) }}
			 	 \weight(H) \dd \vec y_{[n]} \bigg) \notag\\
		& \qquad \times \bigg( \sum_{n \geq 0} \frac{\lambda^n}{n!} \int_{\Lambda^n} \sum_{\substack{H \in \mathcal G(B \cup C,\vec y_{[n]}): \\ \forall i \in [n]: y_i \sim C }} 
						\weight(H) \dd \vec y_{[n]} \bigg) \notag\\
		=& \bigg(\sum_{n \geq 1} \frac{\lambda^n}{n!} \Big(\int_\Lambda \p(A \sim y \sim C, y \nsim B) \dd y \Big)^n \bigg) 
			\bigg(\sum_{n \geq 0} \frac{\lambda^n}{n!} \Big(-\int_\Lambda \p( y \sim (A \cup B)) \dd y \Big)^n \bigg) \notag\\
		& \qquad \times \bigg(\sum_{n \geq 0} \frac{\lambda^n}{n!} \Big(-\int_\Lambda \p( C \sim y \nsim B) \dd y \Big)^n \bigg) . \label{eq:dcfshell:hWL_bound_base}}
Recognizing the exponential series in the expression above, we can rewrite the  probabilities with respect to $\p$ (appearing in the exponents) as probabilities with respect to $\p_\lambda$ associated to $\xi^y$, e.g., $\p( y \sim (A \cup B))= \p_\lambda( y \sim (A \cup B) \text{ in }\xi^y)$. Then we can apply the univariate Mecke's formula (see~\eqref{eq:prelim:mecke_m} for $m=1$) to rewrite~\eqref{eq:dcfshell:hWL_bound_base} as
	\al{ & \Big( \e^{\ela[|\{y \in \eta: A \sim y \sim C, y\nsim B \}|] } - 1\Big) \e^{-\ela[|\{y \in \eta: y \sim (A \cup B) \}|]} \e^{-\ela[|\{y \in \eta: C \sim y \nsim B \}|]}  \\
		=& \Big(1- \e^{-\ela[|\{y \in \eta: A \sim y \sim C, y\nsim B \}|] } \Big) \e^{-\ela[|\{y \in \eta: y \sim (A \cup B) \}|]} \e^{-\ela[|\{y \in \eta: C \sim y \nsim (A \cup B) \}|]} \\ 
		=& (1-Q_{0,k+1})  \e^{-\ela[|\{y \in \eta: y \sim (A \cup B \cup C) \}|]}.  }
Since $\e^{-\ela[|\{y \in \eta: y \sim (A \cup B \cup C) \}|]} = \pla ( \eta \centernot\longleftrightarrow W)$, we can plug this back into~\eqref{eq:dcfshell:hWL_base_split} and obtain
	\[  h(\vec W;L) = \pla ( \eta \centernot\longleftrightarrow W) q_{0,k+1} \Big( 1-Q_{0,k+1} - \p(A\sim C)\Big) \]
on the level of formal power series. Now we prove convergence and check that the previous computational steps are justified not only on the level of formal power series.
We revisit first the Eq.~\eqref{eq:dcfshell:hWL_bound_base}. On the left-hand side, let us 
put absolute values inside the integral (but outside the sum over shell graphs $H$). The resulting expression is bounded by the middle part of~\eqref{eq:dcfshell:hWL_bound_base}, again with absolute values inside the integral. Each integrand is bounded in absolute value by a probability, hence it is smaller or equal to $1$. The resulting series are exponential series and in particular, absolutely convergent. 
As a consequence, Eq.~\eqref{eq:dcfshell:hWL_bound_base} is justified and the last sum in~\eqref{eq:dcfshell:hWL_base_split} is bounded as
	\algn{ \sum_{n \geq 0} \frac{\lambda^n}{n!} &\int_{\Lambda^n} \biggl|
				\sum_{\substack{H \in \mathcal G(W,\vec y_{[n]}): \\ \deg(y_i) \geq 1 \forall i \in[n], \\ \exists i: A \sim y_i \sim C}} \weight(H) \biggr| \dd \vec y_{[n]} \notag \\
		& \leq \e^{\ela[|\{y \in \eta: A \sim y \sim C, y\nsim B \}|] } \e^{\ela[|\{y \in \eta: y \sim (A \cup B) \}|]} \e^{\ela[|\{y \in \eta: C \sim y \nsim B \}|]}  \notag\\
		&\leq \e^{\ela[|\{y \in \eta: y \sim  A\}|] } \e^{\ela[|\{y \in \eta: y \sim (A \cup B) \}|]} \e^{\ela[|\{y \in \eta:  y \sim C\}|]}  \notag\\
		& \leq \e^{2 \lambda|W|},  \label{eq:tbb2}}
where for the last inequality we use the fact that the expected number of direct neighbors of any fixed element of $W$ with respect to $\eta$ is given by $\lambda\int\varphi(x)\dd x$ as well as the re-scaling introduced in Section 2.3 ensuring that $\int\varphi(x)\dd x=1$; compare this bound to the one used in~\eqref{eq:tbb}. For the other contribution to $h(\vec W;L)$, we notice 
	\algn{ & \sum_{n \geq 0} \frac{\lambda^n}{n!} \int_{\Lambda^n}\Biggl|\Big(\sum_{H' \in \mathcal G(A,C): E(H) \neq \varnothing} \weight(H') \Big) 
					\Big(\sum_{\substack{H'' \in \mathcal G(W,\vec y_{[n]}): \\ y_i \sim W \forall i \in[n]}} \weight(H'') \Big) \Biggr| \dd \vec y_{[n]} \notag\\
		&\qquad \leq \p(A\sim C)\, \e^{ \ela[|\{ y \in \eta: y \sim W \}|]}  \leq \e^{\lambda|W|}, \label{eq:tbb3} }
by the same argument as in~\eqref{eq:tbb2}.

Combining~\eqref{eq:tbb2} and~\eqref{eq:tbb3} with~\eqref{eq:dcfshell:hWL_base_split} and $0\leq q_{0,k+1} \leq 1$, we deduce 
	\[ \sum_{n \geq 0} \frac{\lambda^n}{n!} \int_{\Lambda^n} \Bigl| \sum_{H \in \gen{L}: \mathcal S(H) = \vec y_{[n]}}	\weight(H)\Bigr| \dd \vec y_{[n]} 
				\leq \e^{\lambda|W|} + \e^{ 2\lambda|W|} \leq 2\e^{ 2\lambda|W|} < \infty. \] 

\underline{Inductive step.} For the inductive step, let $m>1$. We write the lace $L$ in terms of its vertices $(s_i,t_i)$ in $W$ (that is $s_i=u_{\alpha_i}$ and $t_i=u_{\beta_i}$) and let $L'$ be the lace on $W' := W \setminus [s_1, s_2)$ obtained from $L$ by deleting the first stitch. We note that if $H \in \gen{L}$, then $H[[s_2,u_{k+1}]]\in \gen{L'}$. Observe that
	\eqq{ h(\vec W;L) = h(\vec {W'};L') \sum_{n \geq 0} \frac{\lambda^n}{n!} \int_{\Lambda^n} 
				\sum_{\substack{H \in \mathcal G(\overline V_0, \ldots, \overline V_{\alpha_3-1}, \vec y_{[n]}): \\ H \oplus L' \in \gen{L}}} \weight(H) \dd \vec y_{[n]}. 
						\label{eq:dcfshell_hWL_step_L'_extract}}
Again we first prove~\eqref{eq:dcfshell:hWL_lemma_identity} and carry out computations on the level of formal power series; we prove convergence (and thus~\eqref{eq:dcfshell:hWL_lemma_abs_conv}) at the end.
We can apply the induction hypothesis to $h(\vec {W'};L)$; it remains to deal with the second factor. We partition the vertices in $[s_1, s_3]$ as $A = [s_1]\!], B=(\!(s_1,s_2), C=[s_2,t_1)\!), D = [\![t_1]$, and $E= (t_1,s_3]$ (see Figure~\ref{fig:h_bound_base}). If $m=2$, we let $E=(t_1,u_k]$.

The graphs summed over in~\eqref{eq:dcfshell_hWL_step_L'_extract} must satisfy the following restraints: There must be at least one direct or indirect stitch between $A$ and $D$, and there cannot be any (direct or indirect) edge between $A$ and $E$. In particular, the remaining direct stitches may or may not be there, and thus can be extracted as the factor $q_{\alpha_1, \alpha_2, \alpha_3}$.

We partition $\mathcal S(H) = \cup_{i=1}^4 \mathcal S_i$, where
	\al{\mathcal S_1 &= \{y: A \sim y \sim D, N(y) \subseteq [s_1,t_1] \}, \quad \mathcal S_2 = \{y: A \sim y \sim C, N(y) \subseteq [s_1, t_1)\!) \}, \\
		\mathcal S_3 &= \{y: \varnothing \neq N(y) \subseteq [s_1, s_2)\},  \quad \mathcal S_4 = \{y: B \sim y \sim (C \cup D \cup E), N(y) \subseteq (\!(s_1, s_3]\}.}
Again, we intend to split the sum over graphs into those that have at least one direct stitch between $A$ and $D$, and those that do not. We can thus rewrite the second factor in~\eqref{eq:dcfshell_hWL_step_L'_extract} as
	\algn{ & q_{\alpha_1, \alpha_2, \alpha_3} \sum_{n \geq 0} \frac{\lambda^n}{n!} \int_{\Lambda^n} 
				\sum_{\substack{H \in \mathcal G(\overline V_0, \ldots, \overline V_{\alpha_3-1}, \vec y_{[n]}): \\ H \oplus L' \in \gen{L}, \notag\\
					\forall e\in E(H): e \cap (A\cup D \cup \vec y_{[n]}) \neq \varnothing}} \weight(H) \dd \vec y_{[n]}\\
			=& q_{\alpha_1, \alpha_2, \alpha_3}  \prod_{i=2}^{4}\bigg( \sum_{n_i \geq 0} \frac{\lambda^{n_i}}{n_i!} \int_{\Lambda^{n_i}} 
					\sum_{\substack{H \in \mathcal G(W, \vec y_{i,[n_i]}): \\ \mathcal S(H) = \mathcal S_i}} \weight(H) \dd \vec y_{[i,[n_i]} \bigg) \notag\\
			& \hspace{2cm} \times \bigg[-\p(A\sim D) \sum_{n \geq 0} \frac{\lambda^n}{n!}
					\int_{\Lambda^n}\sum_{\substack{H \in \mathcal G([s_1,s_3],\vec y_{[n]}): \\ y_i \sim A\cup B \forall i \in[n]}} \weight(H) \dd \vec y_{[n]} \notag\\
			& \hspace{3cm} +\sum_{n \geq 1} \frac{\lambda^n}{n!}
				\int_{\Lambda^n}\sum_{\substack{H \in \mathcal G([s_1,s_3],\vec y_{[n]}): \\ y_i \sim A\cup B \forall i \in[n]}} \weight(H) \dd \vec y_{[n]}  \bigg] \notag\\
			=& q_{\alpha_1, \alpha_2, \alpha_3} \big(-\p(A\sim D) \e^{\ela[|\{y \in \eta: y  \in \mathcal S_1\}|]} + \e^{\ela[|\{y \in \eta: y  \in \mathcal S_1\}|]}-1 \big) \notag\\
			& \qquad \times \exp\Big\{\ela[|\{y \in \eta: y  \in \mathcal S_2\}|] 
				- \ela[|\{y \in \eta: y  \in \mathcal S_3\}|] + \ela[|\{y \in \eta: y  \in \mathcal S_4\}|]\Big\}, \label{eq:dcfshell:hWL_step_big_rewrite} }
where the last identity was obtained using Lemma~\ref{lem:laces:weight_sums_probabilities}. Note that the factor $h(\vec{W'};L')$ contains the factor $\p(\eta\centernot\longleftrightarrow [s_2,u_{k+1}]) = \e^{-\ela[|\{ y\in\eta: y \sim [s_2,u_{k+1}]\}|]}$. Together with this factor,~\eqref{eq:dcfshell:hWL_step_big_rewrite} equals
	\algn{& \exp \Big\{ \ela\Big[-|\{y \in \eta: (A \cup B \sim y \sim [s_2,u_{k+1}]\}|+ |\{y \in \eta: A \sim y \sim D, y \nsim (B \cup C)\}| \notag \\
		& \qquad \qquad + |\{y \in\eta: A \sim y \sim C, y\nsim B\}|+|\{y \in \eta: B \sim y \sim (C \cup D \cup E)\}| \Big] \Big\}  \label{eq:dcfshell:hWL_bound_exponent}\\
		& \qquad \times \pla( \eta \centernot\longleftrightarrow W) \big(1-Q_{A,D} - \p(A \sim D) \big). \notag }
It remains to rewrite the argument in the expectation of the exponent in~\eqref{eq:dcfshell:hWL_bound_exponent}. Note that
	\al{ &- |\{y \in \eta: A \sim y \sim [s_2,u_{k+1}], y\nsim B \}| - |\{y \in \eta: B \sim y \sim [s_2,u_{k+1}] \}| \\
		 &+ |\{y \in \eta: A \sim y \sim (C \cup D), y \nsim B\}| + |\{y \in \eta: B \sim y \sim (C \cup D \cup E) \}| \\
		=&  - |\{y \in \eta: A \sim y \sim (t_1,u_{k+1}], y\nsim (B \cup C \cup D) \}| - |\{y \in \eta: B \sim y \sim (s_3,u_{k+1}], y \nsim (C \cup D \cup E) \}|. }
This gives two exponential terms. The first is
	\al{ &\exp\Big\{-\ela[|\{y \in \eta: [u_{\alpha_1}]\!] \sim y \sim (u_{\beta_1},u_{k+1}], y \nsim (\!(u_{\alpha_1},u_{\beta_1}]\}|]\Big\}\\
			=& \prod_{j=\beta_1+1}^{k+1} \exp\Big\{{-\ela[|\{y \in \eta: [u_{\alpha_1}]\!] \sim y \sim [\![u_j], y \nsim (\!(u_{\alpha_1},u_j)\!)\}|]}\Big\}  \\
			=& Q_{\alpha_1,(\beta_1,k+1]} .}
Similarly, the second exponential term equals $Q_{(\alpha_1,\alpha_2), (\alpha_3,k+1]}$. 

Again, we prove convergence and justify the previous computational steps. Revisiting the left-hand side of~\eqref{eq:dcfshell:hWL_step_big_rewrite}, we insert absolute values inside the integral (and outside the sum over graphs $H$). As in the base case, this is bounded by the middle part of~\eqref{eq:dcfshell:hWL_step_big_rewrite} with absolute values in the integrals, and each integrand is a probability. With the Mecke equation, we obtain
	\al{ &\sum_{n \geq 0} \frac{\lambda^n}{n!}  \int_{\Lambda^n} \bigg|
				\sum_{\substack{H \in \mathcal G(\overline V_0, \ldots, \overline V_{\alpha_3-1}, \vec y_{[n]}): \\ H \oplus L' \in \gen{L}, \notag\\
					\forall e\in E(H): e \cap (A\cup D \cup \vec y_{[n]}) \neq \varnothing}} \weight(H) \bigg| \dd \vec y_{[n]} \notag \\
		& \leq 2 \exp\Big\{\ela[|\{y \in \eta: y  \in \mathcal S_1\}|] + \ela[|\{y \in \eta: y  \in \mathcal S_2\}|] 
				+ \ela[|\{y \in \eta: y  \in \mathcal S_3\}|] + \ela[|\{y \in \eta: y  \in \mathcal S_4\}|]\Big\} \notag \\
		& \leq 2 \e^{3\lambda |A \cup B|},}arguing as in~\eqref{eq:tbb2} for the last inequality.

Note that by induction hypothesis, the term $h(\vec{W'};L')$ with absolute values in the respective integrals is bounded by $2^{m-1} \e^{3 \lambda|W'|}$. Since $A\cup B$ and $W'$ are disjoint, this proves~\eqref{eq:dcfshell:hWL_lemma_abs_conv}.
\end{proof}

\begin{proof}[Proof of Proposition~\ref{thm:dcfshell:bounds_on_shell_fct}]
Again, we abbreviate $\eta=\eta_\Lambda$ and $h=h_\lambda^\Lambda$. First, consider $k=0$, i.e., pivot decompositions with no pivotal points. Then there are no direct stitches and
	\[h^{(m)} (\vec W, \vec y_{[m]}) = (-1)^m \prod_{i=1}^{m} \p(y_i \sim W) \qquad \text{and} \qquad h(\vec W) = \pla(\eta \centernot\longleftrightarrow W). \]
Moreover,
	\[ \sum_{m \geq 0} \frac{\lambda^m}{m!} \int_{\Lambda^m} | h^{(m)} (\vec W, \vec y_{[m]}) | \dd\vec y_{[m]} = \e^{\ela[|\{ y \in \eta: y \sim W\}|]} \leq \e^{\lambda|W|}, \]
using the same bound as in~\eqref{eq:tbb3}. Since this proves the proposition for $k=0$, we turn to $k\geq 1$ and we first prove~\eqref{eq:dcfshell:main_bound}.

We rewrite $h(\vec W)$ by explicitly writing out the sum over laces $L$ in terms of the endpoints of their stitches in $W$ (note that any lace can have at most $k$ stitches). We first exhibit this for $k=2$, where $\vec W=(u_0,V_0,u_1,V_1,u_2,V_2,u_3)$ and there are two different laces. With the abbreviation $\tilde Q_{i,j} = Q_{i,j} + \p([u_i]\!] \sim [\![u_j])$,
	\algn{ h(\vec W) &= h(\vec W; L_1) + h(\vec W; L_2) = \pla(\eta \centernot\longleftrightarrow W)
							\Big( q_{0,3} (1-\tilde Q_{0,3}) + Q_{0,3} (1-\tilde Q_{0,2})(1-\tilde Q_{1,3}) \Big)  \notag\\
			&= \pla(\eta \centernot\longleftrightarrow W) \sum_{\beta_1 =2}^{3} q_{0,\beta_1}(1-\tilde Q_{0,\beta_1})
						 \Big[ \mathds 1_{\{\beta_1=3\}} + \mathds 1_{\{\beta_1<3\}} \sum_{\alpha_2=1}^{\beta_1-1} Q_{0,3} (1-\tilde Q_{\alpha_2,3})  \Big]. \label{eq:shell:k2_rewriting}}
Clearly, this is unnecessarily complicated for $k=2$, as the sum over $\alpha_2$ contains only one term and $q_{0,2}=1$. However, this turns out to be convenient for general $k$. We use the convention that $Q_{[a,b],\varnothing} = Q_{\varnothing, [a,b]} =1$. Carefully re-arranging the sum over all laces yields
	\al{ h(\vec W) &= \sum_{L \in \mathcal L(\vec W)} h(\vec W;L) = \pla(\eta\centernot\longleftrightarrow W) \sum_{\beta_1=2}^{k+1} 
					q_{0,\beta_1} (1-\tilde Q_{0,\beta_1}) Q_{0,(\beta_1,k+1]} \\
		& \quad \times \bigg[ \mathds 1_{\{\beta_1=k+1\}} + \sum_{\alpha_2=1}^{\beta_1-1} \sum_{\beta_2=\beta_1+1}^{k+1} 
					q_{\alpha_2, \beta_2}(1-\tilde Q_{\alpha_2,\beta_2}) Q_{(0,\alpha_2],(\beta_2,k+1]} Q_{(0,\alpha_2),\beta_2} \\
		& \quad \times \bigg[ \mathds 1_{\{\beta_2=k+1\}} +  \sum_{\alpha_3=\beta_1}^{\beta_2-1} \sum_{\beta_3=\beta_2+1}^{k+1} 
					q_{\alpha_3, \beta_3}(1-\tilde Q_{\alpha_3,\beta_3}) q_{\alpha_1,\alpha_2,\alpha_3} \\
		&		\hspace{6cm} Q_{(\alpha_2,\alpha_3],(\beta_3,k+1]} Q_{(\alpha_2,\alpha_3),\beta_3} Q_{(0,\alpha_2),(\alpha_3,\beta_2)} \\
		& \quad \times \bigg[ \mathds 1_{\{\beta_3=k+1\}} + \sum_{\alpha_4=\beta_2}^{\beta_3-1} \sum_{\beta_4=\beta_3+1}^{k+1} 
					q_{\alpha_4, \beta_4}(1-\tilde Q_{\alpha_4,\beta_4}) q_{\alpha_2,\alpha_3,\alpha_4} \\
		&		\hspace{6cm} Q_{(\alpha_3,\alpha_4],(\beta_4,k+1]} Q_{(\alpha_3,\alpha_4),\beta_4} Q_{(\alpha_2,\alpha_3),(\alpha_4,\beta_3)} \times \cdots \\
		& \quad \times \Big[ \mathds 1_{\{\beta_{k-1}=k+1\}} +  \mathds 1_{\{\beta_{k-1}<k+1\}} \sum_{\alpha_k=\beta_{k-2}}^{\beta_{k-1}-1}
					q_{\alpha_k, k+1}(1-\tilde Q_{\alpha_k,k+1}) \prod_{j=k,k+1} q_{\alpha_{j-2},\alpha_{j-1},\alpha_j}\\ 
		&		\hspace{6cm} Q_{(\alpha_{k-1},\alpha_k),k+1} Q_{(\alpha_{k-2},\alpha_{k-1}),(\alpha_k,\beta_{k-1})} \Big] \cdots \bigg]\bigg]\bigg].  }
Note that if $\beta=k+1$ for some $i$, then the double sum following the respective indicator breaks down to $0$. Also, only the innermost bracketed term contains two factors of $q_{a,b,c}$.

We now show that, starting with the innermost square brackets, the bracketed terms are bounded by $1$ in absolute value. 

To lighten notation, we write the innermost sum as $\sum_{\alpha=b_1}^{b_2-1} R(\alpha)$. We split the factor $1-\tilde Q_{\alpha_k,k+1} = (1-Q_{\alpha_k,k+1}) - \p([u_{\alpha_k}\!] \sim [\![ u_{k+1}])$. This yields two sums $\sum_{\alpha=b_1}^{b_2-1} R(\alpha) = \sum_{\alpha=b_1}^{b_2-1} R'(\alpha) - \sum_{\alpha=b_1}^{b_2-1} R''(\alpha)$, where $R'$ and $R''$ are both non-negative. Now, with the estimate $Q_{(\alpha_{k-1},\alpha),k+1} \leq Q_{[\beta_{k-2},\alpha),k+1} = Q_{[b_1,\alpha),k+1}$, we can bound
	\algn{  \sum_{\alpha=b_1}^{b_2-1} R'(\alpha) &\leq \sum_{\alpha=b_1}^{b_2-1} (1- Q_{\alpha,k+1}) Q_{[b_1,\alpha),k+1} \notag \\
		&= (1-Q_{b_1,k+1}) + Q_{b_1,k+1}  \sum_{\alpha=b_1+1}^{b_2-1} (1- Q_{\alpha,k+1}) Q_{[b_1+1,\alpha),k+1}, \label{eq:shell:Q_innermost_bracket_pos}}
which is readily proven to be at most $1$ by induction. Moreover,
	\eqq{ \sum_{\alpha=b_1}^{b_2-1} R''(\alpha) \leq  \sum_{\alpha=b_1}^{b_2-1} q_{\alpha,k+1} \p([u_\alpha]\!] \sim [\![ u_{k+1}]). \label{eq:shell:Q_innermost_bracket_neg}}
The above summands can be rewritten as the probability of the event that the direct stitch $(\alpha,k+1)$ is present, while all direct stitches $(j,k+1)$ for $j\in (\alpha,k+1]$ are not. Hence, these are disjoint events for different values of $\alpha$, and so the sum is at most $1$.

In total, we rewrote $\sum_{\alpha=b_1}^{b_2-1} R(\alpha)$ as the difference of two nonnegative values, both at most $1$, proving our claim.

To deal with the summands for $2 \leq i <k$, we write the double sum as $\sum_{\alpha=b_1}^{b_2-1} \sum_{\beta = b_2+1}^{k+1} R(\alpha, \beta)$ and split the term $1-\tilde Q_{\alpha_i,\beta_i} = (1- Q_{\alpha_i,\beta_i}) - \p([u_{\alpha_i}\!] \sim [\![ u_{\beta_i}])$ so that
	\eqq{ \sum_{\alpha=b_1}^{b_2-1} \sum_{\beta = b_2+1}^{k+1} R(\alpha, \beta) 
				= \sum_{\alpha=b_1}^{b_2-1} \sum_{\beta = b_2+1}^{k+1} R'(\alpha, \beta) 
				- \sum_{\alpha=b_1}^{b_2-1} \sum_{\beta = b_2+1}^{k+1} R''(\alpha, \beta) \label{eq:shell:Q_general_bracket_split}}
for non-negative summands $R', R''$. We prove a bound on the sum over $R'(\alpha,\beta)$ by induction on $k-b_2$. If $b_2=k$, then the bound is the same as for the bound~\eqref{eq:shell:Q_innermost_bracket_pos}. For $b_2<k$, we first bound $Q_{(\alpha_i,\alpha],(\beta,k+1]} \leq Q_{[b_1,\alpha],(\beta,k+1]}$ and then extract the summand for $\beta=k+1$, yielding
	\algn{\sum_{\alpha=b_1}^{b_2-1} \sum_{\beta = b_2+1}^{k+1} R'(\alpha, \beta) & \leq \sum_{\alpha=b_1}^{b_2-1} \sum_{\beta = b_2+1}^{k+1}
					(1-Q_{\alpha,\beta}) Q_{[b_1,\alpha],(\beta,k+1]} Q_{[b_1,\alpha),\beta} \notag\\
			&\leq \sum_{\alpha=b_1}^{b_2-1} (1-Q_{\alpha,k+1}) Q_{[b_1,\alpha),k+1} \notag\\
			& \qquad + Q_{[b_1,b_2-1],k+1} \sum_{\alpha=b_1}^{b_2-1} \sum_{\beta = b_2+1}^{k} (1-Q_{\alpha,\beta}) Q_{[b_1,\alpha],(\beta,k]} Q_{[b_1,\alpha),\beta}. 
							\label{eq:shell:Q_general_bracket_induction}}
By the induction hypothesis, the double sum in~\eqref{eq:shell:Q_general_bracket_induction} is at most $1$. Therefore,
	\al{\sum_{\alpha=b_1}^{b_2-1} \sum_{\beta = b_2+1}^{k+1} R'(\alpha, \beta) & \leq \sum_{\alpha=b_1}^{b_2-2} (1-Q_{\alpha,k+1}) Q_{[b_1,\alpha),k+1} \\
		& \qquad + (1-Q_{b_2-1,k+1}) Q_{[b_1,b_2-1),k+1} + Q_{b_2-1,k+1} Q_{[b_1,b_2-1),k+1} \\
		&=  \sum_{\alpha=b_1}^{b_2-2} (1-Q_{\alpha,k+1}) Q_{[b_1,\alpha),k+1} + Q_{[b_1,b_2-2],k+1} \\
		&= 1,  }
where the last identity is now an easy induction.

Turning to the second summand in~\eqref{eq:shell:Q_general_bracket_split}, by the same argument used to treat~\eqref{eq:shell:Q_innermost_bracket_neg}, the summands $R''(\alpha,\beta)$ are probabilities of events which are disjoint for different values of $(\alpha,\beta)$, and so they sum to at most $1$.

The observation that the bracket term for $i=1$ is handled analogously finishes the proof of~\eqref{eq:dcfshell:main_bound}.

We proceed to prove~\eqref{eq:dcfshell:m_shell_summability} for $k>1$. By combining Lemma~\ref{lem:laces:size_L} with Lemma~\ref{lem:laces:shell_fct_fixed_lace},
	\al{ \sum_{m\geq 0} \frac{\lambda^m}{m!} \int_{\Lambda^m} \big| h^{(m)}(\vec W,\vec y_{[m]}) \big| \dd\vec y_{[m]}
		& \leq \sum_{L \in \mathcal L_k}\sum_{m \geq 0} \frac{\lambda^m}{m!} 
				\int_{\Lambda^m} \Big| \sum_{H \in \gen{L}: \mathcal S(H) = \vec y_{[m]}} \textup{\weight}(H) \Big| \dd \vec y_{[m]} \\
		& \leq \frac{1}{\sqrt{5}} \Big( \frac{3+\sqrt{5}}{2} \Big)^{k} 2^k \e^{3\lambda|W|}.  }
Using the bound $k \leq |W|$ finishes the proof.
\end{proof}

\begin{lemma}[Thermodynamic limit of the shell function] \label{lem:hlimit} 
For every $\lambda \geq 0$, the pointwise limit 
	\[ \lim_{\Lambda \nearrow \R^d} h_\lambda^\Lambda(\vec W) = h_\lambda(\vec W) 	\] 
along $\R^d$-exhausting sequences exists. 
\end{lemma} 

\begin{proof}
Let $(\Lambda_n)_{n\in\N}$ be an $\Rd$-exhausting sequence. For fixed $\vec W=(u_0, V_0, \ldots, u_{k+1})$, note that
	\[ 	h_\lambda^{\Lambda_n}(\vec W) = \sum_{L\in\mathcal L_k} h_\lambda^{\Lambda_n}(\vec W;L). 	\]
For each lace $L$, the limit
	\[	h_\lambda(\vec W;L) = \lim_{n\to \infty} h_\lambda^{\Lambda_n}(\vec W;L)	\] 	
exists and does not depend on the precise choice of $\R^d$-exhausting sequence. This is clear from the representation for $h_\lambda^\Lambda(\vec W;L)$ proven in Lemma~\ref{lem:laces:shell_fct_fixed_lace}: In particular, $h_\lambda^\Lambda(\vec W;L)$ is given as the finite product of $\Lambda$-independent factors and factors that describe the probability of certain point processes containing no points (namely, $\pla(\eta_\Lambda \centernot\longleftrightarrow W)$ and the factors $Q_{i,j}$). As probabilities that are decreasing in the volume, the latter admit a $\Lambda \nearrow \Rd$ limit. It follows that the limit of the shell function exists as well and is given by 
	\[	h_\lambda(\vec W) = \sum_{L\in\mathcal L_k} h_\lambda(\vec W;L).   	\]
\end{proof}

\subsection{The direct-connectedness function in infinite volume} \label{sec:dcf_inf}
In this section, we consider the limit $\lim_{\Lambda \nearrow \Rd} \dcf^{\Lambda}$ with $\dcf^\Lambda$ as in~\eqref{eq:def:dcf} and give sufficient conditions under which it exists, thereby proving the two convergence statements from Theorem~\ref{thm:main_thm}.

The candidate limit is given by the analogue of~\eqref{eq:def:dcf} with $\Lambda$ replaced by $\R^d$; the existence of $h_\lambda^{\R^d}\equiv h_\lambda$ has been checked in Lemma~\ref{lem:hlimit}. Thus, 
	\eqq{ g_\lambda(x_1,x_2) = \sum_{r=0}^\infty \frac{\lambda^r}{r!} \int_{(\R^d)^r} \sum_{\vec W} 
				\Bigl( \sum_{G\in\mathcal G^{\vec W}_{\text{core}}} \weight^\pm(G)\Bigr) h_\lambda(\vec W) \dd \vec x_{[3,r+2]}, \label{eq:gfinal} }
where the inner sum is over core graphs $G$ on $\vec x_{[r+2]}$ with pivot decomposition $\vec W$, i.e., over $(+)$-connected graphs $G$ on $\vec x_{[r+2]}$ with $\textsf{PD}^+(x_1,x_2,G)= \textsf{PD}^\pm(x_1,x_2,G)=\vec W$. Remember the quantities $0<\tilde \lambda_*\leq \lambda_*$ introduced before Theorem~\ref{thm:main_thm}. We will see in~\eqref{eq:dcf:inf_vol:rewrite_core_probabilities} that the sum over core graphs for a given pivot decomposition is a probability, hence in particular non-negative.

\begin{theorem}[The thermodynamic limit of $\dcf^\Lambda$: pointwise convergence.] \label{thm:dcf_inf:limit_theorem}
If $\lambda< \lambda_\ast$, then 
	\[	\sum_{r=0}^\infty \frac{\lambda^r}{r!} \int_{(\R^d)^r}  \sum_{\vec W} \Bigl( \sum_{G\in\mathcal G^{\vec W}_{\mathrm{core}}} \textup{\weight}^\pm(G)\Bigr) \bigl| h_\lambda(\vec W)\bigr| \dd \vec x_{[3,r+2]}<\infty 	\]
for all $x_1,x_2\in \R^d$. Moreover, for every $\R^d$-exhausting sequence $(\Lambda_n)_{n\in \N}$, we have the pointwise convergence 
	\[ \lim_{n\to \infty} g_\lambda^{\Lambda_n} (x_1,x_2) = g_\lambda(x_1,x_2) 	\]
with $g_\lambda$ given in~\eqref{eq:gfinal} (equivalently, Eq.~\eqref{eq:def:dcf} with $\Lambda$ replaced by $\Rd$).
\end{theorem}

\begin{theorem}[Integrability and convergence in the $L^1$-norm] \label{thm:dcf_inf:limit_theorem:L1}
If $\lambda< \tilde \lambda_\ast$, then for all $x_1\in \R^d$, 
	\[ \int_{\R^d} |g_\lambda(x_1,x_2)|\dd x_2\leq \sum_{r=0}^\infty\frac{\lambda^r}{r!} \int_{\R^d}\Bigl(  \int_{(\R^d)^r}  \sum_{\vec W} 
				\Bigl( \sum_{G\in\mathcal G^{\vec W}_{\mathrm{core}}} \textup{\weight}^\pm(G)\Bigr) \bigl|h_\lambda(\vec W)\bigr| \dd \vec x_{[3,r+2]} \Bigr) \dd x_2 <\infty. \]
\end{theorem}

\begin{proof} [Proof of Theorem~\ref{thm:dcf_inf:limit_theorem}]
	We consider a summand in~\eqref{eq:gfinal} for fixed $\vec W$ and set $x_1=u_0$ as well as $x_2=u_{k+1}$. Let $\vec W=(u_0,V_0, \ldots, V_k, u_{k+1})$. Remember $\overline{V}_i = \{u_{i}\}\cup V_i\cup \{u_{i+1}\}$. A first important observation is the fact that the weight of a core graph with pivot decompotion $\vec W$ factorizes into the product over the $k$ $(\pm)$-subgraphs induced by the vertex sets $\overline{V}_i$. The sum over core graphs thus factorizes as
	\algn{ \hspace{-0.5cm}\sum_{\substack{G \in \mathcal C^+(W): \\ \textsf{PD}^+(G)= \textsf{PD}^\pm(G)=\vec W}} \hspace{-0.5cm}\weight^\pm(G) 
			&=  \prod_{i=0}^{k} \Big( \sum_{H \in \mathcal D^+_{u_i,u_{i+1}}(\overline V_i)} \weight^\pm(H) \Big) \notag\\
			&= \prod_{i=0}^{k} \p \big(\rg(\overline V_i) \in \mathcal D_{u_i, u_{i+1}} \big) \notag\\
			&=  \p \Big(\bigcap_{i=0}^k \{\rg(\overline V_i) \in \mathcal D_{u_i, u_{i+1}} \} \Big). \label{eq:dcf:inf_vol:rewrite_core_probabilities} }
Hence, the core can be written as a probability. Combining this with Proposition~\ref{thm:dcfshell:bounds_on_shell_fct}, we get
	\begin{align*}
	 \Bigl( \hspace{-0.5cm}\sum_{\substack{G \in \mathcal C^+(W): \\ \textsf{PD}^+(G)=\textsf{PD}^\pm(G)=\vec W}} \hspace{-0.5cm}  \weight^\pm(G)\Bigr) \bigl| h_\lambda^\Lambda(\vec W)\bigr|
		&\leq \pla(\eta_\Lambda \centernot\longleftrightarrow W)\p \Big(\bigcap_{i=0}^k \{\rg(\overline V_i) \in \mathcal D_{u_i, u_{i+1}} \} \Big)\\
		&= \pla \Big( \{ \C(u_0, \xi^W_\Lambda ) = W \} \cap \bigcap_{i=0}^k \{\xi^W_\Lambda\big[\overline V_i\big] \in \mathcal D_{u_i, u_{i+1}} \} \Big).
	\end{align*} 
	Above, we used independence as well as the fact that for $V \subseteq W$, the two random graphs $\rg(V)$ and $\xi^W[V]$ are identical in distribution. The inequality holds true for bounded $\Lambda$ as well as $\Lambda = \R^d$. 
	
We now go back to~\eqref{eq:def:dcf} and re-arrange the sum by first summing over the number of pivotal points $k$, giving
	\algn{	& \sum_{r=0}^\infty \frac{\lambda^r}{r!} \int_{\Lambda^r}  \sum_{\vec W} \Bigl( \sum_{G\in\mathcal G^{\vec W}_{\text{core}}} \textup{\weight}^\pm(G)\Bigr)
					 \bigl| h_\lambda^\Lambda(\vec W)\bigr| \dd \vec x_{[3,r+2]} \notag \\
	= & \sum_{k \geq 0} \lambda^k \sum_{n \geq 0} \frac{\lambda^n}{n!} \int_{\Lambda^{k+n}} \sum_{\vec W}
				\Big( \sum_{G\in\mathcal G^{\vec W}_{\text{core}}} \weight^\pm (G) \Big) |h_\lambda^\Lambda(\vec W)| \dd \vec v_{[n]} \dd\vec u_{[k]} \label{eq:dcf:inf_vol:sum_by_pivotals}}
In the second term, the sum is over pivot decompositions $\vec W=(u_0, V_0, \ldots, V_k, u_{k+1})$ where $u_0=x_1, u_{k+1} =x_2$, and $\cup_{i=0}^k V_i = \{v_1, \ldots, v_n\}$.

When rewriting the integrand of~\eqref{eq:dcf:inf_vol:sum_by_pivotals} as a probability, the event that $u_i$ and $u_{i+1}$ are $2$-connected for $i\in [k]_0$ in disjoint vertex sets $V_i$ becomes the event that these connection events occur disjointly within $W$, see Section~\ref{sec:defs} and recall definition \eqref{def:odot}. The inner series can thus be bounded as
\algn{ & \sum_{n \geq 0} \frac{\lambda^n}{n!} \int_{\Lambda^{n}} \sum_{\vec W} \Big( \sum_{G\in\mathcal G^{\vec W}_{\text{core}}} \weight^\pm (G) \Big) |h_\lambda^\Lambda(\vec W)| \dd \vec v_{[n]}\notag \\
	\leq & \sum_{n \geq 0} \frac{\lambda^{n}}{n!} \int_{\Lambda^{n}} 
					\pla \Big( \big\{\C(u_0, \xi_\Lambda^{\vec u_{[k]}, \vec v_{[n]}}) = \vec u_{[k]} \cup \vec v_{[n]} \big\} \notag \\
			& \hspace{2cm} \cap\Big( \{\dconn{u_0}{u_1}{\xi^{u_0,u_1,\vec v_{[n]}}}\} \circ \cdots
					\circ \{\dconn{u_k}{u_{k+1}}{\xi^{u_k,u_{k+1},\vec v_{[n]}}} \} \Big) \Big) \dd \vec v_{[n]} \notag \\
	 = & \pla\big( \{\dconn{u_0}{u_1}{\xi^{u_0,u_1}}\} \circ \cdots \circ \{\dconn{u_k}{u_{k+1}}{\xi^{u_k,u_{k+1}}} \}\big),  \label{eq:dcf:inf_vol:disj_occurrence}} 
where the identity is due to the Mecke equation and due to the fact that by summing over $\vec v$, we were partitioning over what the joint cluster of $\vec u_{[0,k+1]}$ is. We can now use the BK inequality (\cite[Theorem 2.1]{HeyHofLasMat19}) to bound~\eqref{eq:dcf:inf_vol:disj_occurrence} by
	\eqq{ \prod_{i=0}^{k} \pla \big(\dconn{u_i}{u_{i+1}}{\xi_\Lambda^{u_i,u_{i+1}}} \big) \leq  \prod_{i=0}^{k} \dtlam(u_{i+1}-u_i). \label{eq:dcf:inf_vol:BK} }
Inserting this back into~\eqref{eq:dcf:inf_vol:sum_by_pivotals},
	\algn{ &\sum_{k \geq 0} \lambda^k \sum_{n \geq 0} \frac{\lambda^n}{n!} \int_{\Lambda^{k+n}} \sum_{\vec W}
				\Big( \sum_{G\in\mathcal G^{\vec W}_{\text{core}}} \weight^\pm (G) \Big) |h_\lambda^\Lambda(\vec W)| \dd \vec v_{[n]} \dd\vec u_{[k]}  \notag\\
		\leq & \sum_{k \geq 0} \lambda^k \dtlam^{\ast (k+1)}(x_2-x_1). \label{eq:dcf:inf_vol:convolution_bound} }
The last expression is finite for $\lambda<\lambda_\ast$, by the definition of $\lambda_\ast$. The pointwise convergence of $\dcf^{\Lambda_n}$ to $\dcf$ follows by dominated convergence.
\end{proof}

\begin{proof}[Proof of Theorem~\ref{thm:dcf_inf:limit_theorem:L1}]	
If we integrate over $x_2$ in~\eqref{eq:dcf:inf_vol:convolution_bound}, this yields the upper bound
	\[ \lambda^{-1} \sum_{k \geq 1} \Big( \lambda \int\dtlam(x) \dd x \Big)^k, \]
which is finite for $\lambda<\tilde\lambda_\ast$, by definition of $\tilde\lambda_\ast$.
  The theorem follows by Fubini-Tonelli and the triangle inequality. 
\end{proof}

\section{The Ornstein-Zernike equation} \label{sec:oze}

Here we complete the proof of Theorem~\ref{thm:main_thm}. In view of Theorems~\ref{thm:dcf_inf:limit_theorem} and~\ref{thm:dcf_inf:limit_theorem:L1}, it remains to prove that the expansion \eqref{eq:gfinal} is indeed equal to the direct-connectedness function given by the OZE \eqref{eq:intro:oze_2d}. This is proven by showing first that $g_\lambda^\Lambda$ from Definition \ref{def:dcfshell:shell_dcf} fulfills the OZE in finite volume and then passing to the limit $\Lambda\nearrow \R^d$. 

The proof idea in finite volume is basically well known;  the same proof works for the Ornstein-Zernike equation for the total correlation function. 

\begin{prop}[The Ornstein-Zernike equation in finite volume] \label{thm:OZE_finite}
Let $\Lambda \subset \Rd$ be bounded and let $x_1, x_2 \in \Lambda$. Then
	\[ \tlam^\Lambda(x_1, x_2) = \dcf^\Lambda(x_1, x_2) + \lambda \int_\Lambda \dcf^\Lambda(x_1, x_3) \tlam^\Lambda(x_3, x_2) \dd x_3.\]
\end{prop}
\begin{proof}
We drop the $\Lambda$-dependence in the superscript of $\tlam^\Lambda$ and $\dcf^\Lambda$. Thanks to Proposition~\ref{thm:tlam_finite_cluster_expansion}, we can re-sum the series expansion for $\tlam$ at will. Given a pivot decomposition $\vec W = (u_0, V_0, \ldots, u_{k+1})$ of an arbitrary core graph $G$ with the vertex set $W$, define
	\eqq{\label{def:h_bar} \bar h_\lambda^{(m)}(\vec W, \vec y_{[m]}) := \hspace{-0.5cm}\sum_{\substack{H \in \mathcal G(\overline V_1, \ldots, \overline V_k, \vec y_{[m]}): \\
				 G \oplus H \in \mathcal C_{u_0,u_{k+1}}^\pm(W \cup \vec y_{[m]})}} \hspace{-0.5cm}\normalfont{\weight}(H), \qquad
				  \bar h_\lambda(\vec W)  := \sum_{m \geq 0} \frac{\lambda^m}{m!} \int_{\Lambda^m} \bar h_\lambda^{(m)}(\vec W, \vec y_{[m]}) \dd \vec y_{[m]} }  
in analogy to the shell function $h_\lambda$ in~\eqref{eq:def:shell_fct} (just as the latter, $\bar{h}_\lambda$ only depends on $G$ through its pivot decomposition $\vec W$). To be more precise, the shell function $h_\lambda$ is recovered from $\bar{h}_\lambda$ by summing over a smaller subset of graphs $H$ in~\eqref{def:h_bar}, adding the restriction that $G\oplus H$ shall not contain  ($\pm$)-pivot points for the $u_0$-$u_{k+1}$-connection. Note that 
	\[ 0 \leq \bar h_\lambda(\vec W) = \e^{-\ela[|\{y \in \eta: y\sim W\}|]} \prod_{x,y \in W: \nexists i \in [k]_0: \{x,y\} \subseteq \bar V_i} (1-\connf(x-y))  \]
and that when replacing $h_\lambda$ with $\bar h_\lambda$ in the right-hand side of~\eqref{eq:def:dcf}, we get $\tlam$ instead of $\dcf$. We can split the sum $\bar h_\lambda^{(m)}(\vec W, \vec y_{[m]}) = h_\lambda^{(m)}(\vec W, \vec y_{[m]}) + f_\lambda^{(m)} (\vec W, \vec y_{[m]})$, where $f_\lambda^{(m)}$ contains the sum over those graphs $H$ so that $G \oplus H$ \emph{does} have $(\pm)$-pivotal points with respect to the $u_0$-$u_{k+1}$-connection. We set $f_\lambda(\vec W)  := \sum_{m \geq 0} \frac{\lambda^m}{m!} \int_{\Lambda^m} f_\lambda^{(m)}(\vec W, \vec y_{[m]}) \dd \vec y_{[m]}$.

Assume now that $u_j$ for $j \in [k]$ is the first pivotal point of $G \oplus H\in \mathcal C_{x_1,x_2}^\pm(W \cup \vec y_{[m]})$. Furthermore, let $\vec W_j' := (u_0, V_0, \ldots, u_j)$, let $\vec W_j'':=(u_j, V_j, \ldots, u_{k+1})$, and let $y_{[s]}$ for $s \leq m$ be the points adjacent to $\vec W_j'$ (possibly after reordering the vertices). The weight of such a graph $H$ then factorizes into the product of the weights of two graphs, namely the subgraphs of $H$ induced by $\vec W'_j\cup\vec y_{[s]}\subset V(H)$ and by $\vec W''_j\cup \vec y_{[m]\setminus [s]}\subset V(H)$, see Figure~\ref{fig:h_factorizes_decomp}. That is,
\begin{align}
\weight(H) = \weight(H[\vec W_j'\cup\vec y_{[s]}]) \weight(H[\vec W''_j\cup \vec y_{[m]\setminus [s]}]).\label{eq:h_decomp_factorization}	
\end{align}
Moreover, we see that $H[\vec W_j'\cup\vec y_{[s]}]\oplus G[W_j']$ does not contain $(\pm)$-pivot points (for the $u_0$-$u_j$-connection) and $H[\vec W_j''\cup\vec y_{[m]\setminus[s]}]\oplus G[W_j'']$ is in general just $(\pm)$-connected.
\begin{figure}
\centering
\includegraphics[scale=1.0]{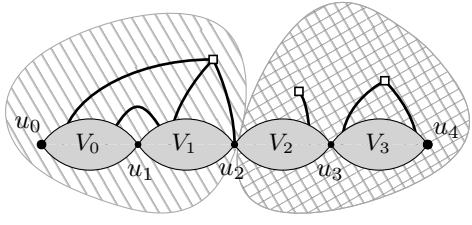}
\caption{The schematic representation of a ($\pm$)-graph $G\oplus H$ in $\mathcal C_{u_0,u_4}^\pm(W \cup \vec y_{[3]})$ illustrates the factorization of the graph weight from Equation~\eqref{eq:h_decomp_factorization}: The edges of $H$ are explicitly depicted in the picture, while the core graph $G$ is represented by its pivot decomposition $(u_0,V_0,\ldots,u_4)$. The vertices $\vec y_{[3]}$ are depicted by squares, ordered from left to right. The first ($\pm$)-pivot point for the $u_0$-$u_4$-connection in $G\oplus H$ is $u_2$. Thus, the weight of the simple graph $H$ factorizes into the weight of its subgraph induced by $\vec W'_2\cup \vec y_{[1]}=\{u_0,u_1,u_2\}\cup V_0\cup V_1\cup\{y_1\}$ (hatched on the left) and the weight of the subgraph induced by $\vec W''_2\cup\vec y_{[3]\backslash [1]}=\{u_3,u_4\}\cup V_2\cup V_3\cup\{y_2,y_3\}$ (crosshatched on the right).\label{fig:h_factorizes_decomp}} 
\end{figure}

By partitioning over $j$, we thus obtain the decomposition
	\[ f_\lambda(\vec W) = \sum_{j=1}^{k} h_\lambda(\vec W_j') \bar h_\lambda(\vec W_j'').\]
Since both $h_\lambda$ and $\bar h_\lambda$ converge absolutely, so does $f_\lambda$, justifying all re-summations. Letting $x_1=u_0$ and $x_2= u_{k+1}$,
	\al{ (\tlam-\dcf)(x_1,x_2) &= \sum_{k \geq 1} \lambda^k \sum_{n_0, \ldots, n_k \geq 0}  \frac{\lambda^{\sum_{i=0}^{k}n_i}}{\prod_{i=0}^{k} n_i!} 
							\int_{\Lambda^{k+\sum_{i=0}^{k} n_i}} \Big( \prod_{i=0}^{k} \p\big(\rg(\bar V_i) \in \mathcal D_{u_i, u_{i+1}}(\bar V_i) \big)\Big) \\
			& \hspace{4cm} \times \Big( \sum_{j=1}^{k} h_\lambda(\vec W_j') \bar h_\lambda(\vec W_j'') \Big) \prod_{i=0}^{k}\dd \vec v_{i,[n_i]} \dd \vec u_{[k]}  \\
			& = \sum_{j \geq 1,k \geq 0} \lambda^{j+k} \int_{\Lambda} \sum_{n_0, \ldots, n_{j+k} \geq 0}  \frac{\lambda^{\sum_{i=0}^{j+k}n_i}}{\prod_{i=0}^{j+k} n_i!}
							\bigg[ \int_{\Lambda^{j-1+\sum_{i=0}^{j-1} n_i}} \prod_{i=0}^{j-1} \p\big(\rg(\bar V_i) \in \mathcal D_{u_i, u_{i+1}}(\bar V_i) \big)\Big) \\
			& \hspace{6cm} \times h_\lambda(\vec W_j') \prod_{i=0}^{j-1}\dd \vec v_{i,[n_i]} \dd \vec u_{[j-1]} \bigg] \\
			& \hspace{3cm} \times \bigg[\int_{\Lambda^{k+\sum_{i=j}^{j+k} n_i}} \prod_{i=j}^{j+k} \p\big(\rg(\bar V_i) \in \mathcal D_{u_i, u_{i+1}}(\bar V_i) \big)\Big) \\
			& \hspace{6cm} \times \bar h_\lambda(\vec W_j'') \prod_{i=j}^{j+k}\dd \vec v_{i,[n_i]} \dd \vec u_{[j+1,j+k]} \bigg] \dd u_j \\
			& = \lambda \int_\Lambda \dcf(x_1,u) \tlam(u,x_2) \dd u.  }
The re-summation w.r.t.~$j$ and $k$ is justified as the resulting series converges for $\lambda<\lambda_\ast$ even when putting $h_\lambda$ in absolute values.
\end{proof}

We can now extend the result of Proposition~\ref{thm:OZE_finite} to $\Lambda \nearrow \Rd$ and thus prove that the expansion \eqref{eq:gfinal} is indeed equal to the the direct-connectedness function for $\lambda<\lambda_\ast$, finalizing the proof of our main result. 

\begin{proof}[Proof of Theorem~\ref{thm:main_thm}]
We have
	\eqq{ \tlam(x_1,x_2) =  \lim_{\Lambda \nearrow \Rd} \tlam^\Lambda(x_1, x_2) = \lim_{\Lambda \nearrow \Rd} \dcf^\Lambda(x_1,x_2)
			 + \lambda  \lim_{\Lambda\nearrow\Rd}  \int_{\Rd}  \dcf^\Lambda(x_1,x_3) \mathds 1_{\Lambda}(x_3) \tlam^\Lambda(x_3, x_2)  \dd x_3, \label{eq:OZE_infinite}}
where the first equality holds by the continuity of probability measures along sequences of increasing events and the second one by Proposition \ref{thm:OZE_finite}.

Note that that integrand in~\eqref{eq:OZE_infinite} is bounded uniformly in $\Lambda$ by
	\[C \tlam(x_3,x_2) , \]
where $C= \sup_{y\in\Rd} \sum_k \lambda^k \sigma_\lambda^{\ast (k+1)}(y)$ is a constant obtained in~\eqref{eq:dcf:inf_vol:convolution_bound}. Since $\tlam$ is integrable for all $\lambda<\lambda_c$, the theorem follows with dominated convergence.
\end{proof}

\section{Discussion} \label{sec:connections}
\subsection{Connections to percolation on Gibbs point processes} \label{sec:connections_physics}

The Ornstein-Zernike equation gets its name from the seminal paper~\cite{OrnZer14} and has since been a well-known formalism in liquid-state statistical mechanics. It relates the total correlation function to the direct correlation function and it naturally connects to power-series expansions of these correlation functions (see~\cite{ConDeAngFor77, Ste76, Ste96}; the terminology is not the same in all of these references). 

The correlation functions admit graphical expansions that also consist of connected graphs. It was observed~\cite{Hil55} that a similar formalism can be formulated for the pair-connectedness function, and a key reference for this is~\cite{ConDeAngFor77}. The pair-connectedness function is deemed part of the pair-correlation function. The connected graphs appearing in the expansion of the latter are referred to as ``mathematical clusters'', and they correspond to our $(\pm)$-connected graphs. Isolating the $(+)$-connected components within these graphs yields the ``physical clusters'', and the graphs in which $x_1$ and $x_2$ lie in the same physical cluster make up the expansion for $\tlam(x_1,x_2)$. In the following, we elaborate on this.

The percolation models considered in the physics literature are mostly not based on a Poisson point process (Stell calls the Poisson setup \emph{random percolation}~\cite{Ste96}), but on a Gibbs point process (called \emph{correlated percolation} in the language of Stell). (The denomination ``random percolation'' for the Poisson setup feels quite misleading for probabilists; but it reflects language commonly adopted across physics, with ``random'' understood as ``completely random'' in the sense of \emph{completely random measures}~\cite{kingman1967}, a class comprising the Poisson point process.)

To define the latter, consider a non-negative pair potential $v\colon \Rd \to \R_{\geq 0}$ and some finite volume $\Lambda$. Let $\mathbf N(\Lambda)$ be the set of finite counting measures on $\Lambda$ and let $\mu \in\mathbf N(\Lambda)$. Then the energy of $\{x_1, \ldots, x_n\}$ under boundary condition $\mu$ is
	\[ H(\{x_1,\ldots, x_n\} \mid \mu) = \sum_{1 \leq i <j \leq n} v(x_i-x_j) + \sum_{i=1}^{n} \sum_{y \in \mu} v(x_i-y).\]
Let $f\colon \mathbf N(\Lambda) \to \R$ be bounded. We define a probability measure as
	\[ \E_z[f] := \frac{1}{\Xi(z)} \sum_{n \geq 0} \frac{z^n}{n!} \int_{\Lambda^n} f(\{x_1, \ldots, x_n\}) \e^{-H(\{x_1, \ldots, x_n\}) } \dd \vec x_{[n]}, \]
where the partition function $\Xi(z)$ is so that $\E_z[1] = 1$ and $z \in \R_{\geq 0}$ is called the activity. If we denote by $\eta$ a random variable with law $\E_z$, then $\eta$ is a point process. Note that we recover the homogeneous PPP with intensity $\lambda=z$ by setting $v \equiv 0$.

We can define the RCM $\xi$ on this general point process and we denote its probability measure by $\p_{z,\connf}$. We furthermore define the (one-particle) density as
	\[ \rho_1(x) = z \E_z[\e^{-H(\{x\} \mid \eta)} ] = \rho,\]
as well as the pair-correlation function as
	\[ \rho_2(x,y) = z^2 \E_z[\e^{- H(\{x,y\} \mid \eta)}]. \]
Again, in case of a homogeneous PPP with intensity $\lambda=z$, we have $\rho=z$ and $\rho_2= z^2$. Defining the pair-connectedness function as
	\[ \tau_{z,\connf}(x,y) := \E_{z,\connf} \big[\e^{-H(\{x,y\} \mid \eta)} \mathds 1_{\{\conn{x}{y}{\xi^{x,y}}\}}\big],\]
we can decompose 
\[
	\rho_2(x,y) = z^2 \tau_{z,\connf}(x,y) + z^2 \E_{z,\connf} [\e^{-H(\{x,y\} \mid \eta)} \mathds 1_{\{\nconn{x}{y}{\xi^{x,y}}\}}].
\]	
In~\cite{ConDeAngFor77}, Coniglio et al.~define the pair-connectedness function as $\tilde\tau_{z,\connf} = (z^2 /\rho^2) \tau_{z,\connf}$.

The function $\tilde \tau_{z,\connf}$ has a density expansion (note that $\tau_{z,\connf}$ is better suited for activity expansions) that can be found in~\cite[Equation~(12)]{ConDeAngFor77}, which can be obtained from the density expansion of the pair-correlation function: The latter is obtained by expanding the Mayer-$f$ functions $f(x,y) = \e^{-v(x,y)} - 1$ in the partition function, which is the starting point of a cluster expansion. Splitting the Mayer-$f$ function as $f=f^+ + f^\ast$ with $f^+ = \e^{-v(x,y)} \connf(x-y)$ and executing the same expansion for the correlation function ``doubles'' every edge into a $(+)$-edge and a $(\ast)$-edge. Only summing over graphs in which $x$ and $y$ are connected by $(+)$-edges yields the pair-connectedness function.

In general, the graphs appearing in the density expansion are a subset of those in the activity expansion, namely the ones without articulation points (articulation points were defined after Proposition~\ref{thm:tlam_finite_cluster_expansion}). In the case of a homogeneous PPP, we have $\lambda=z=\rho$, and so both activity and density expansion coincide (and the graphs with articulation points cancel out). Moreover, $f^+(x,y) = - f^\ast(x,y) = \connf(x-y)$ and the graphs summed over in the expansion become the $(\pm)$-graphs, yielding the expansion~\eqref{eq:tlam_cluster_exp} for $\tlam$.

It is an interesting question which ideas of this paper can be generalized to RCMs based on Gibbs point processes. And while some aspects generalize without much effort, the crucial difference lies in the fact that the weight of graphs showing up in expansions for Gibbs point processes also encodes the pair interaction induced by the potential $v$. To recover probabilistic interpretations for terms after performing re-summations and bounds is therefore much more delicate.

\subsection{Connections to Last and Ziesche} \label{sec:connections_lz}
In~\cite{LasZie17}, Last and Ziesche use a Margulis-Russo-type formula to prove analyticity of $\tlam$ in presumably the whole subcritical regime. Moreover, they show the existence of some $\lambda_0>0$ (which is not quantified) so that both $\tlam$ and $\dcf$ have an absolutely convergent graphical expansion in $[0,\lambda_0)$ that seems closely related to the ones discussed here. We want to illustrate how to relate the respective expressions.

\paragraph{The two-point function.} Last and Ziesche show that $\tlam(x_1,x_2)$ is equal to
	\eqq{ \sum_{n\geq 0} \frac{\lambda^n}{n!} \int \sum_{J \subset [3,n+2]} (-1)^{n-|J|} \p\big(\conn{x_1}{x_2}{\rg(\vec x_{J\cup\{1,2\}})},
			\rg(\vec x_{[n+2]}) \text{ is connected}\big) \dd \vec x_{[3,n+2]}. \label{eq:links:lz:tlam_exp}}
We show that the above integrand is the same as the one in~\eqref{eq:tlam_cluster_exp}. We can rewrite the one in~\eqref{eq:links:lz:tlam_exp} as
	\algn{ \E\Big[ \mathds 1_{\{\rg(\vec x_{[n+2]}) \text{ is connected}\}} 
				\sum_{J \subset [n+2]} (-1)^{n-|J|} \mathds 1_{\{\conn{x_1}{x_2}{\rg(\vec x_{J})}\}} \Big]. \label{eq:links:lz:tlam_rewrite}}
Note that now, any non-vanishing $J$ needs to contain $\{1,2\}$. We are now going to observe some cancellations. For a fixed graph $G\in\mathcal C(\vec x_{[n+2]})$,
	\algn{\sum_{J \subseteq [n+2]} (-1)^{n-|J|} \mathds 1_{\{\conn{x_1}{x_2}{G[\vec x_{J}]}\}} 
				&= \sum_{I,J \subseteq [n+2]} (-1)^{n-|J|} \mathds 1_{\{\{1,2\} \subseteq I \subseteq J\}} \mathds 1_{\{\C(x_1, G[\vec x_{J}]) = \vec x_{I}\}} \notag\\
			&= \sum_{I,J \subseteq [n+2]} (-1)^{n-|J|} \mathds 1_{\{\{1,2\} \subseteq I \subseteq J\}} \mathds 1_{\{G[ \vec x_{I}] \text{ connected}\}}
						\mathds 1_{\{\forall j\in J\setminus I: x_j \nsim \vec x_{I}\}}  \notag\\
			&=\sum_{\{1,2\} \subseteq I\subseteq [n+2]} (-1)^{n-|I|} \mathds 1_{\{G[ \vec x_{I}] \text{ connected}\}} \notag\\
			& \hspace{2cm} \times \sum_{J \subseteq [n+2] \setminus I} (-1)^{|J|} \mathds 1_{\{\forall j\in J: x_j \nsim \vec x_{I}\}} . 
							\label{eq:links:lz:integrand_rewrite}}
Note that for given $G$ and $I$, defining $\mathcal I(G,I) = \{ j \in [n+2] \setminus I:  x_j \nsim \vec x_{I} \}$, we can rewrite 
	\eqq{ \sum_{J \subseteq [n+2] \setminus I} (-1)^{|J|} \mathds 1_{\{\forall j\in J: x_j \nsim \vec x_{ I}\}}
			= \sum_{J \subseteq \mathcal I(G,I)} (-1)^{|J|} . \label{eq:links:lz:J_cancellations}}
The only case for which~\eqref{eq:links:lz:J_cancellations} does not vanish is when $\mathcal I(G,I) = \varnothing$. We can therefore rewrite~\eqref{eq:links:lz:integrand_rewrite} as
	\al{ \sum_{J \subseteq [n+2]} (-1)^{n-|J|} & \mathds 1_{\{\conn{x_1}{x_2}{G[\vec x_{ J}]}\}}  = \sum_{\{1,2\} \subseteq I\subseteq [n+2]}
			 (-1)^{n-|I|} \mathds 1_{\{G[ \vec x_{I}] \text{ connected}\}}  \mathds 1_{\{\forall j\in [n+2]\setminus I: x_j \sim \vec x_{I}\}}	}
and so~\eqref{eq:links:lz:tlam_rewrite} becomes
	\al{ & \sum_{\{1,2\} \subseteq I \subseteq [n+2]} (-1)^{n-|I|}\p\big( \rg(\vec x_{I}) \text{ is connected}, \ 
				x_j \sim \vec x_{I} \forall j \in [n+2]\setminus I \big)   \\
			=& \sum_{\{1,2\} \subseteq I \subseteq [n+2]} \p\big( \rg(\vec x_{I}) \text{ is connected} \big) 
				\prod_{j \in [n+2]\setminus I} \big[\prod_{i\in I} (1-\connf(x_i-x_j)) -1 \big] \\
			=& \sum_{I \subseteq [n+2]} \sum_{G} \weight^\pm (G).}
In the last line, summation is over the same set of graphs as in~\eqref{eq:rcmfin:tlam_exp_one_sum} with the additional restriction that $V(G^+) =I$. Resolving the partition over $I$ gives that~\eqref{eq:links:lz:tlam_exp} is equal to~\eqref{eq:rcmfin:tlam_exp_one_sum}.

\paragraph{The direct-connectedness function.} In~\cite[Theorem 5.1]{LasZie17}, it is shown that there exists $\lambda_0$ such that for $\lambda \in [0, \lambda_0)$,
	\algn{ \dcf(x_1,x_2) &= \sum_{n \geq 0} \frac{\lambda^n}{n!} \int \sum_{G \in \mathcal D_{x_1,x_2}(\vec x_{[n+2]})} \prod_{e \in E(G)} \connf(e) \notag \\
				& \qquad \qquad \times \sum_{\substack{H \in \mathcal C(\vec x_{[n+2]}): \\ H \subseteq G }} 
					\sum_{\substack{J \subseteq [n+2]: \\ x_1 \longleftrightarrow x_2 \text{ in } H[\vec x_J]}} (-1)^{n-|J| + |E(G) \setminus E(H)|} \dd\vec x_{[3,n+2]}. 
							\label{eq:links:lz:dcf_lz}}
We show that the integrand in~\eqref{eq:links:lz:dcf_lz} is equal to the one in~\eqref{eq:dcf_rewriting}. With the calculations~\eqref{eq:links:lz:integrand_rewrite} and~\eqref{eq:links:lz:J_cancellations} performed for the two-point function, letting $I^c = [n+2]\setminus I$,
	\[ \sum_{J \subseteq [n+2]} (-1)^{n-|J|} \mathds 1_{\{x_1 \longleftrightarrow x_2 \text{ in } H[\vec x_J]\}}
			= \sum_{\{1,2\} \subseteq I \subseteq [n+2]} (-1)^{n-|I|} \mathds 1_{\{H[\vec x_I] \text{ is connected}\}} 
							\mathds 1_{\{\forall j\in I^c: x_j \sim \vec x_{I} \text{ in } H\}}.\]
The two indicators imply that $H$ is connected, and so
	\algn{&\sum_{\substack{H \in \mathcal C(\vec x_{[n+2]}): \\ H \subseteq G }} 
						\sum_{\substack{J \subseteq [n+2]: \\ x_1 \longleftrightarrow x_2 \text{ in } H[\vec x_J]}} (-1)^{n-|J| + |E(G) \setminus E(H)|} \notag \\
		=& \sum_{\{1,2\}\subseteq I \subseteq [n+2]} \sum_{H \subseteq G} (-1)^{n-|I| + |E(G) \setminus E(H)|} \mathds 1_{\{H[\vec x_I] \text{ is connected}\}} 
							\mathds 1_{\{\forall j\in I^c: x_j \sim \vec x_{I} \text{ in } H\}} \notag \\
		=& \sum_{\{1,2\}\subseteq I \subseteq [n+2]} \sum_{\substack{H' \in \mathcal C(\vec x_I): \\ H' \subseteq G}} (-1)^{n-|I| + |E(G) \setminus E(H')|}
							\sum_{\substack{F \subset E(G) \cap ( (I\times I^c) \cup \binom{I^c}{2}): \\ \forall j\in I^c: F \cap (I \times \{j\}) \neq \varnothing}} (-1)^{|F|}.  
								\label{eq:links:lz:dcf_rewrite} }
Note that for the second identity in~\eqref{eq:links:lz:dcf_rewrite}, we split the edges of $H$ into those contained in $H'$ (the subgraph induced by $I$) and the remaining ones, called $F$.

When $E(G) \cap \binom{I^c}{2} \neq \varnothing$, then the sum over $F$ vanishes. Hence, the sum over $I$ can be reduced to those $I$ such that $G[I^c]$ contains no edges. For such sets $I$, we have
	\eqq{ \sum_{\substack{F \subset E(G) \cap (I\times I^c): \\ \forall j\in I^c: F \cap (I \times \{j\}) \neq \varnothing}} (-1)^{|F|}
			= \prod_{j\in I^c} \sum_{\varnothing \neq F_j \subseteq E(G) \cap (I \times \{j\} )} (-1)^{|F_j|} = \prod_{j\in I^c} (-1) = (-1)^{n-|I|}. \label{eq:links:lz:dcf:cancellation}}
When inserting~\eqref{eq:links:lz:dcf:cancellation} back into~\eqref{eq:links:lz:dcf_rewrite}, the two factors $(-1)^{n-|I|}$ cancel out, and so
	\algn{&\sum_{\substack{H \in \mathcal C(\vec x_{[n+2]}): \\ H \subseteq G }} 
						\sum_{\substack{J \subseteq [n+2]: \\ x_1 \longleftrightarrow x_2 \text{ in } H[\vec x_J]}} (-1)^{n-|J| + |E(G) \setminus E(H)|}  \notag \\
		=& \sum_{\substack{\{1,2\}\subseteq I \subseteq [n+2]: \\ E(G[I^c]) = \varnothing }} \sum_{\substack{H \in \mathcal C(\vec x_I): \\ H \subseteq G}} (-1)^{|E(G) \setminus E(H)|} \notag\\
		=& \sum_{H \in \mathcal G(\vec x_{[n+2]}): G \vartriangleright H} \mathds 1_{\{\{1,2\} \subseteq V(H)\}} (-1)^{|E(G) \setminus E(H)|}, \label{eq:links:lz:dcf_rewrite_b}}
where $G \vartriangleright H$ means that $E(H) \subseteq E(G)$, the subgraph of $H$ induced by the vertices incident to at least one edge (call this set $V_{\geq 1}(H)$) is connected, and the subgraph of $G$ induced by $[n+2]\setminus V_{\geq 1}(H)$ contains no edges.

With the identity~\eqref{eq:links:lz:dcf_rewrite_b}, and letting $X=\vec x_{[n+2]}$, the integrand of~\eqref{eq:links:lz:dcf_lz} is equal to
	\algn{ &\sum_{\substack{H \in \mathcal G(X): \\ \{x_1,x_2\} \subseteq V_{\geq 1}(H), \\ H[V_{\geq 1}(H)] \text{ connected}}} \prod_{e \in E(H)} \connf(e)
				\sum_{\substack{F \subseteq \binom{X}{2} \setminus E(H): \\ \forall e \in F: e \cap V_{\geq 1}(H) \neq \varnothing,\\ (X, F \cup E(H)) \in \mathcal D_{x_1,x_2}(X) }} 
						(-1)^{|F|} \prod_{e \in F} \connf(e) \notag\\
		=& \sum_{\substack{H \in \mathcal G(X): \\ x_1 \longleftrightarrow x_2}} \prod_{e \in E(H)} \connf(e)
				\sum_{\substack{F \subseteq \binom{X}{2} \setminus E(H): \\ (X, F \cup E(H)) \in \mathcal D_{x_1,x_2}(X) }} 
						(-1)^{|F|} \prod_{e \in F} \connf(e) \notag\\
		=& \sum_{\substack{C \in \mathcal D^\pm_{x_1,x_2}(X): \\ x_1 \overset{+}{\longleftrightarrow} x_2}} \weight^\pm (G). \label{eq:links:lz:dcf_final_rewrite}}
The argument for the first identity in~\eqref{eq:links:lz:dcf_final_rewrite} is the same for the identity of~\eqref{eq:rcmfin:tlam_exp_one_sum} and~\eqref{eq:rcmfin:tlam:final_exp}.

\subsection{Connections to the lace expansion} \label{sec:connections_hhlm}

Both the graphical power-series expansions and the lace expansion provide an expression for the direct-connectedness function. In this section, we show how to get from one to the other. Note that the statements to follow hold for sufficiently small intensities and can \emph{not} replace the lace expansion, which works all the way up to $\lambda_c$. The emphasis of this section is on the qualitative nature of the results.

We first sum up some results of~\cite{HeyHofLasMat19}, where the lace expansion is applied to the RCM. We keep some of the definitions brief and informal, and we point to~\cite{HeyHofLasMat19} for the detailed definitions in these cases.

\paragraph{On the lace expansion.} In~\cite{HeyHofLasMat19}, among other things, the Ornstein-Zernike equation is proved for $\tlam$ in \emph{high dimension} (and for certain classes of connection functions $\connf$, see~\cite[Section 1.2]{HeyHofLasMat19}). In particular, it is shown that
	\[ \dcf(x) = \connf(x)  + \Pi_\lambda(x), \]
with $\Pi_\lambda(x) = \sum_{n \geq 0} (-1)^n \Pi_\lambda^{(n)}(x)$. The functions $\Pi_\lambda^{(n)}$ are called the \emph{lace-expansion coefficients}, they are non-negative, and have a quite involved probabilistic interpretation. To briefly define them, let $\{\xconn{x}{y}{\xi^{x,y}}{A}\}$ be the event that $\conn{x}{y}{\xi^{x,y}}$, but $x$ is no longer connected to $y$ in an $A$-thinning of $\eta^y$. Informally, every point $z\in \eta$ survives an $A$-thinning with probability $\prod_{y\in A} (1-\connf(z-y))$. See~\cite[Definition~3.2]{HeyHofLasMat19} for a formal definition. Letting
	\[ E(x,y;A,\xi^{x,y}) = \{ \xconn{x}{y}{\xi^{x,y}}{A}\} \cap \{ \nexists w \in \piv{x,y;\xi^{x,y}}: \xconn{x}{w}{\xi^x}{A} \},\]
we introduce a sequence $\xi_0, \ldots, \xi_n$ of independent RCMs and define
	\algn{ \Pi_\lambda^{(0)}(x) &:= \sigma_\lambda(x) - \connf(x), \notag\\
		\Pi_\lambda^{(n)}(u_n) &:= \lambda^n \int \pla\Big( \{\dconn{\orig}{u_0}{\xi_0^{\orig, u_0}}\} \cap  \bigcap_{i=1}^n
					E\big(u_{i-1}, u_i; \C(u_{i-2}, \xi_{i-1}^{u_{i-2}}), \xi_i^{u_{i-1}, u_i}\big) \Big) \dd \vec u_{[0,n-1]} \label{eq:def:links:le_coefficients}}
for $n\geq 1$ (with $u_{-1}=\orig$). The method of proof is called the \emph{lace expansion}, a perturbative technique that first proves via induction that
	\eqq{ \tlam(x) = \connf(x) + \sum_{m=0}^{n} (-1)^m\Pi_\lambda^{(m)}(x) + \lambda \Big(\Big(\connf + \sum_{m=0}^{n} (-1)^m\Pi_\lambda^{(m)} \Big)\ast \tlam \Big)(x) + R_{\lambda,n}(x) 
				\label{eq:links:tlam_LE_n_identity}}
for $n \in \N_0$ and some remainder term $R_{\lambda,n}$ (see~\cite[Definition~3.7]{HeyHofLasMat19}), and then shows that the partial sum converges to $\Pi_\lambda=\dcf-\connf$ and that $R_{\lambda,n}\to 0$ as $n \to\infty$.

The lace expansion was first devised for self-avoiding walk by Brydges and Spencer~\cite{BrySpe85} and takes some inspiration from cluster expansions. It was later applied to percolation (specifically, bond percolation on $\Zd$) by Hara and Slade~\cite{HarSla90}. While the name stems from laces that appear in the pictorial representation in~\cite{BrySpe85}, laces are absent in the representation for percolation models.

We show that we can rewrite $\Pi_\lambda^{(n)}$ in terms of graphs that are associated to a lace of size $n$. More generally, rewriting $\Pi_\lambda^{(n)}$ should serve as a bridge between the graphical expansions for $\dcf$ that are well known in the physics literature, and the expression for $\dcf$ in terms of lace-expansion coefficients.

The big advantage in the lace expansion lies in the probabilistic nature of all appearing terms, allowing to bound most appearing integrals by the expected cluster size, which is finite for $\lambda<\lambda_c$. The downside is the absence of a direct expression of $\dcf$ and thus a direct proof of the OZE, which is only obtained after performing the $n\to\infty$ limit in~\eqref{eq:links:tlam_LE_n_identity}.

We now show how to re-sum the graphical expansion for $\tlam$ and how to obtain the lace-expansion coefficients by appropriate grouping of terms. 

\paragraph{Building the connection.} For $x,y\in X$, let $\tilde{\mathcal C}_{x,y}^\pm(X) \subset \mathcal C^\pm(X)$ be the set of graphs in $\mathcal C^\pm(X)$ such that $G^+$ is connected and contains $\{x,y\}$, and $E(G[ V \setminus V^+]) = \varnothing$. Hence, all $(-)$-edges are incident to at least one vertex in $V(G^+)$. This is exactly the set of graphs summed over in~\eqref{eq:rcmfin:tlam_exp_one_sum}. Indeed,
	\eqq{\tlam(x_1, x_2) = \sum_{n\geq 0} \frac{\lambda^n}{n!} \int_{(\Rd)^n} \sum_{G \in \tilde{\mathcal C}_{x_1, x_2}^\pm(\vec x_{[n+2]})} \weight^\pm(G) \dd \vec x_{[3,n+2]}.
			\label{eq:def:links:tlam_exp}}
If we define $\tilde{\mathcal D}_{x,y}^\pm(X) := \mathcal D_{x,y}^\pm(X) \cap \tilde{\mathcal C}_{x,y}^\pm(X)$, we can express $\dcf(x_1,x_2)$ by replacing the graphs summed over in~\eqref{eq:def:links:tlam_exp} by $\tilde{\mathcal D}_{x,y}^\pm(\vec x_{[n+2]})$.

We are going to recycle some notation from Section~\ref{sec:dcfshell}. We split $G$ into its core $G_{\text{core}}$ and its shell $H$, so that $\textsf{PD}^+(x,y,G_{\text{core}}) = \textsf{PD}^\pm(x,y,G_{\text{core}})= (u_0, V_0, u_1, \ldots, u_k, V_k, u_{k+1})$ for some $k$ (where $u_0=x$ and $u_{k+1}=y$). We also recall that $G$ ``contains" a skeleton (see Definition~\ref{def:laces:skeletons}), a graph on $[k+1]_0$.

\begin{definition}[The minimal lace] \label{def:links::minlace}
Let $G $ be a graph with core $G_{\text{core}}$ and shell $H$; let $\vec W= (u_0,V_0 \ldots, u_{k+1})$ for $k \in \N$ be its $(+)$-pivot decomposition. We define the \emph{minimal lace} $L_{\text{min}}(x,y;G)$ as lace with the following properties:
\begin{compactitem}
\item $L$ (having bonds $\alpha_i\beta_i$ with $i \in [m]$ for some $m\in\N$) is contained as a subgraph in the skeleton $\hat H$;
\item for every $i \in [m]$, among all the bonds $\alpha\beta$ in $\hat H$ satisfy $\alpha<\beta_{i-1}$, the bond $\alpha_i\beta_i$ maximizes the value of $\beta$. For $i=1$, we take $\beta_{0}=1$.
\end{compactitem}
If $\textsf{Piv}^+(x,y;G)= \varnothing$, we say that $G$ has a minimal lace of size $0$.
\end{definition}
In other words, the first stitch $0\beta_1$ maximizes the value of $\beta_1$ among all stitches starting at $0$, the second stitch has a maximal value of $\beta_2$ among the stitches with $1 \leq \alpha_2 < \beta_1$, and so on.

As a side remark, it is worth noting that the minimal laces offer an alternative way of partitioning the set of all shell graphs by mapping every shell graph $H$ onto its minimal lace. This gives a standard procedure used in lace expansion for self-avoiding walk; performing it ``backwards'' yields precisely the mapping described below Definition~\ref{def:laces}.

With the notion of minimal laces, we partition
	\eqq{\dcf(x_1,x_2) = \sum_{m \geq 0}\pi_\lambda^{(m)}(x_1,x_2), \label{eq:links:dcf_little_pis}} 
where
	\eqq{ \pi_\lambda^{(m)}(x_1,x_2) :=  \sum_{n\geq 0} \frac{\lambda^n}{n!} \int_{(\Rd)^n}
					\sum_{\substack{G \in \tilde{\mathcal D}_{x_1, x_2}^\pm(\vec x_{[n+2]}): \\ \| L_{\text{min}}\|= m}} \weight^\pm(G) \dd \vec x_{[3,n+2]}. 
							\label{eq:def:links:little_pis}}
We also set $\pi_\lambda^{(m)}(x) = \pi_\lambda^{(m)}(\orig,x)$. 

We strongly expect that the (pointwise) absolute convergence of the power series on the right-hand side of \eqref{eq:def:links:little_pis} holds (at least) in the domain of absolute convergence of the physicists' expansion \eqref{eq:dcf_rewriting} and thus, as already discussed, for sufficiently small intensities $\lambda>0$. However, a proof would go beyond the scope of the discussion here, therefore we formulate the absolute convergence of $\pi_\lambda^{(m)}$ (in the above sense) as an assumption for the following result (Lemma \ref{lem:links:Pi_pi_identity}).

\begin{assumption}\label{ass:1} There exists $0<\lambda_\bigstar\leq\lambda_c$ such that the right-hand side of \eqref{eq:def:links:little_pis} is (pointwise) absolute convergent for all $m\in\N$ and $\lambda<\lambda_\bigstar$.
\end{assumption}

Under Assumption~\ref{ass:1}, we show that the coefficients defined in \eqref{eq:def:links:little_pis} are basically identical to the lace-expansion coefficients introduced in \eqref{eq:def:links:le_coefficients}.  

\begin{lemma}[Identity for the lace-expansion coefficients] \label{lem:links:Pi_pi_identity}
Let $m \geq 1$ and let $\lambda<\lambda_\bigstar$. Then
	\al{\Pi_\lambda^{(0)}(x) &= \pi_\lambda^{(0)}(\orig,x) - \connf(x), \\
		(-1)^m\Pi_\lambda^{(m)}(x) &= \pi_\lambda^{(m)}(\orig,x).}
\end{lemma}

As a side note, since $\Pi_\lambda^{(m)}$ is non-negative, Lemma~\ref{lem:links:Pi_pi_identity} shows that the sign of $\pi_\lambda^{(m)}$ alternates, which is far from obvious from the definition in~\eqref{eq:def:links:little_pis}.

Next, we prove an approximate version of the OZE in analogy to~\cite[Proposition 3.8]{HeyHofLasMat19}. Clearly, Lemma~\ref{lem:links:oze_pis} follows immediately from the latter via Lemma~\ref{lem:links:Pi_pi_identity}; however, we want to present a short independent proof on the level of formal power series, which we consider instructive for the understanding of the underlying combinatorics. We emphasize that the proof presented here treats the claim of Lemma~\ref{lem:links:oze_pis} as an identity between formal power series; in particular, we do not concern ourselves with absolute convergence of the power series appearing in \eqref{eq:links:tlam_pis_partition} and in \eqref{eq:links:tlam_pis_partition2}.
\begin{lemma}[The lace expansion in terms of $(\pm)$-graph coefficients] \label{lem:links:oze_pis}
Let $m \in \N_0$, let $\lambda<\lambda_\bigstar$, and set $\pi_{\lambda,m}(x):= \sum_{i=0}^{m} \pi_\lambda^{(i)}(\orig,x)$. Then
	\[ \tlam(x) = \pi_{\lambda,m}(x) + \big( \pi_{\lambda,m} \ast \tlam \big)(x) + R_{\lambda,m}(x), \]
where $R_{\lambda,m}$ is defined in~\cite[Definition~3.7]{HeyHofLasMat19}.
\end{lemma}

Before performing the proof of Lemma~\ref{lem:links:Pi_pi_identity}, we define
	\[ \bar\connf(A,B) = \prod_{a\in A} \prod_{b\in B} (1-\connf(a-b))\]
and $\bar\connf(a,B) = \bar\connf(\{a\},B)$. Now, observe that, given a set $A \subset \Rd$ and a RCM event $F$,
	\eqq{ \sum_{n \geq 0} \frac{\lambda^n}{n!} \int_{(\Rd)^n} \bar\connf(A, \vec v_{[n+2]}) \sum_{\substack{G \in \tilde{\mathcal C}_{v_1,v_2}^\pm(\vec v_{[n+2]}): \\ G^+ \in F}} 
			\weight^\pm(G) \dd \vec v_{[3,n+2]} = \pla ( \xi(\thinning{\eta^{v_1,v_2}}{A}) \in F),  \label{eq:links:obs_thinning}}
where $\xi(\eta)$ is the RCM on the basis of the point process $\eta$ and $\thinning{\eta^v}{A}$ is an $A$-thinning of $\eta^v$ (the usual PPP of intensity $\lambda$ and added point $v$). In particular, $v$ may be thinned out as well. We remark that $\thinning{\eta}{A}$ has the same distribution as a PPP of intensity $\lambda\bar\connf(A,\cdot)$.

\begin{proof}[Proof of Lemma~\ref{lem:links:Pi_pi_identity}]
The statement for $m=0$ is clear. For $m>0$, we can rewrite $\pi_\lambda^{(m)}$ as
	\eqq{ \pi_\lambda^{(m)}(u_{-1}, u_m) = \lambda^m \int_{(\Rd)^m} \sum_{k,n \geq 0} \frac{\lambda^{k+n}}{k!n!} \int_{(\Rd)^{k+n}} \sum_{G \in \mathcal B} \weight^\pm(G) 
				\dd \big(\vec u_{[0,m-1]}, \vec x_{[k]}, \vec z_{[n]} \big),  \label{eq:links:pi_m_rewrite}}
where $\mathcal B \subseteq \tilde{\mathcal D}_{u_{-1},u_m}^\pm( \vec u_{[-1,m]} \cup \vec x_{[k]} \cup \vec z_{[n]} )$ are the graphs such that
\begin{enumerate}
	\item[$\bullet$] $u_0$ is the first pivotal point in $\textsf{Piv}^+(u_{-1}, u_m;G)$ (i.e., $\ord(u_0)=2$);
	\item[$\bullet$] $ \vec u_{[0,m-1]} \subseteq \textsf{Piv}^+(u_{-1},u_{m};G)$ and $u_{i-1} \prec u_{i}$;
	\item[$\bullet$] there are points $p_2, \ldots, p_{m}$ such that $L_{\text{min}} \mathrel{\hat =}\{(u_{-1},u_1),(p_2, u_2), \ldots, (p_{m}, u_{m})\}$;
	\item[$\bullet$] $\vec z_{[n]}$ are those vertices $z\notin \{ u_{m-1}, u_m\} $ in $G$ so that $\{z\} \cup N(z)$ contains at least one vertex $y$ of order $y \succ u_{m-1}$.
\end{enumerate}

\begin{figure}
	\centering
 \includegraphics[scale=0.98]{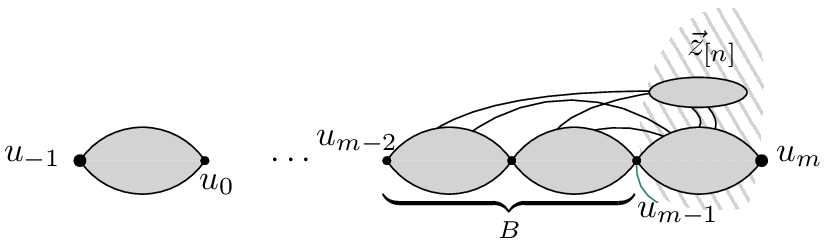}  \includegraphics[scale=0.98]{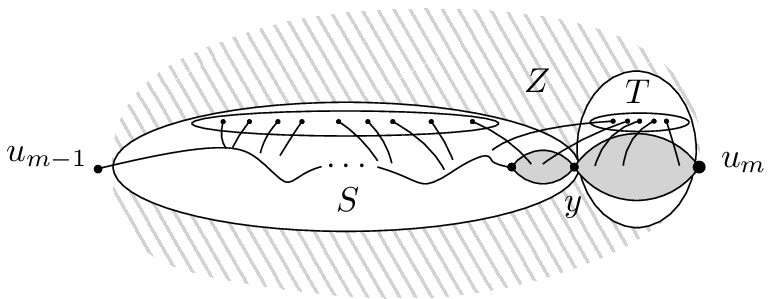}
	\caption{Illustration for the proof of Lemma~\ref{lem:links:Pi_pi_identity}. On the left, we see an example graph $G \in \mathcal B$; the grey bags on the bottom represent $\textsf{PD}^+(u_{-1},u_m,G)$ (note that there can be pivotal points within a grey bag). The minimal lace $L_{\text{min}}$ is not illustrated; however, note that $p_m \in B$. On the right, we see a schematic zoom into $G[Z]$, where $Z = \vec z_{[n]}$, together with the partition $Z=S \cup T$.}
	\label{fig:links_pi_(m)}
\end{figure}

Given a graph $G \in \mathcal B$, let $B$ denote the set of points $x$ in $V(G^+)$ with $u_{m-2} \preccurlyeq x \prec u_{m-1}$. See Figure~\ref{fig:links_pi_(m)} for an illustration of such a graph $G$. We integrate out the points $\vec z$ first and claim that their contribution to~\eqref{eq:links:pi_m_rewrite} is
	\eqq{ \lambda \sum_{n \geq 0} \frac{\lambda^{n}}{n!} \int_{(\Rd)^n} \sum_{H \in \mathcal B^\star} \weight^\pm(H) \dd\vec z_{[n]} 
			= -\lambda\pla\big(E(u_{m-1}, u_m;B ,\xi^{u_{m-1},u_m})\big),  \label{eq:links:pi_m_proof_goal}}
where every $H \in \mathcal B^\star$ is the subgraph of some $G \in \mathcal B$ and has vertex set $B \cup \{u_{m-1}, u_m\} \cup \vec z_{[n]}$ and precisely those edges in $G$ that have at least one endpoint in $\{u_m\} \cup \vec z_{[n]}$.

We let $y$ be the last pivotal point in $V(G^+)$, that is, $\ord(y) = \ord(u_m)-2$. We write $Z=\vec z_{[n]}$ and split $Z$ once more into those vertices ``in front'' and ``behind'' $y$; that is, $Z= S \cup T$, where $T$ are the points in $G^+$ of order $\ord(u_m)-1$ together with the points in $V\setminus V(G^+)$ that are adjacent to the former, and $S = Z \setminus T$. Possibly $y=u_{m-1}$, in which case $S = \varnothing$. See Figure~\ref{fig:links_pi_(m)} for an illustration of this split of the vertices in $Z$.

Note that there are no restrictions on the $(-)$-edges between $B$ and $S \cup \{y\}$, whereas there must be at least one $(-)$-edge between $B$ and $T \cup \{u_m\}$. There are no restrictions on the $(-)$-edges between $\{u_{m-1}\} \cup S \cap V(G^+)$ and $T \cup \{u_m\}$, whereas there cannot be any $(-)$-edges between $S\setminus V(G^+)$ and $T \cup \{u_m\}$. By distinguishing whether or not $S = \varnothing$, the left-hand side of~\eqref{eq:links:pi_m_proof_goal} is equal to
	\algn{ & \sum_{n \geq 0} \frac{\lambda^n}{n!} \int_{(\Rd)^n} \big( \bar\connf(B,\vec z_{[n]} \cup \{u_m\} ) - 1\big) 
				\sum_{\substack{G \in \tilde{\mathcal C}_{u_{m-1},u_m}^\pm (\{u_{m-1},u_m\} \cup \vec z_{[n]}): \\
							 u_{m-1} \overset{+}{\Longleftrightarrow} u_m }} \weight^\pm(G) \dd \vec z_{[n]}  \notag\\
		+ & \lambda \sum_{k \geq 0} \frac{\lambda^k}{k!} \int_{(\Rd)^{k+1}} \bar\connf(B,\vec s_{[k]} \cup \{y\}) 
				\sum_{\substack{H \in \tilde{\mathcal C}_{u_{m-1},y}^\pm (\{u_{m-1},y\} \cup \vec s_{[k]}): \\
							 u_{m-1} \overset{+}{\longleftrightarrow} y }} \weight^\pm(H) \notag\\
		 & \qquad \times \bigg( \sum_{n \geq 0} \frac{\lambda^n}{n!} \int_{(\Rd)^n} \big( \bar\connf(B,\vec t_{[n]} \cup \{u_m\} ) - 1\big)
		 					 \bar\connf(V^+(H) \setminus \{y\},\vec t_{[n]} \cup \{u_m\}) \notag\\
		& \qquad\ \qquad	\sum_{\substack{G \in \tilde{\mathcal C}_{y,u_m}^\pm (\{y,u_m\} \cup \vec t_{[n]}): \\
							 y \overset{+}{\Longleftrightarrow} u_m }} \weight^\pm(G) \dd \vec t_{[n]} \bigg) \dd \vec s_{[k]} \dd y \notag\\
		=& \pla (\dconn{u_{m-1}}{u_m}{\xi(\{u_{m-1}\} \cup \thinning{\eta^{u_m}}{B})}) - \pla (\dconn{u_{m-1}}{u_m}{\xi^{u_{m-1}, u_m}}) \notag\\
		 & \quad +  \lambda\int_{\Rd} \ela \Big[ \mathds 1_{\{\conn{u_{m-1}}{y}{\xi(\thinning{\{u_{m-1}\}\cup\eta^{y}}{B})\}}} \notag\\
		 & \hspace{2cm} \times \Big( \pla \big(\dconn{y}{u_m}{\xi(\{y\} \cup \thinning{\eta^{u_m}}{B\cup \C})}\big) 
		 					- \pla \big(\dconn{y}{u_m}{\xi(\{y\}\cup \thinning{\eta^{u_m}}{\C})}\big)  \Big) \Big] \dd y,  \label{eq:links:pi_m_prob_bridge}}
where we abbreviated $\C=\C(u_{m-1}, \xi^{u_{m-1}})$. Note that the inner probabilities are conditional on the random variable $\C$. We now resolve the integral over $y$ by use of the Mecke equation and incorporate the first two summands as the case $y=u_{m-1}$.
With this,~\eqref{eq:links:pi_m_prob_bridge} becomes
	\algn{ \ela\Big[& \sum_{y \in \eta^{u_{m-1}}}\mathds 1_{\{\conn{u_{m-1}}{y}{\xi(\thinning{\{u_{m-1}\}\cup\eta^{y}}{B})\}} }
											\mathds 1_{\{\dconn{y}{u_m}{\xi(\{y\} \cup\thinning{(\eta^{u_{m}} \setminus\C')}{B})\}}}\Big] \notag \\  
		-& \ela\Big[ \sum_{y \in \eta^{u_{m-1}}}\mathds 1_{\{\conn{u_{m-1}}{y}{\xi^{u_{m-1}}}\}} \mathds 1_{\{\dconn{y}{u_m}{\xi(\eta^{u_m} \setminus \C')\}}}\Big], \label{eq:links:pi_m_bridge2}}
where $\C' = \C(u_{m-1}, \xi(\eta^{u_{m-1}} \setminus\{y\}))$. But both terms in~\eqref{eq:links:pi_m_bridge2} are simply a partition over the last pivotal point for the connection between $u_{m-1}$ and $u_m$, and so~\eqref{eq:links:pi_m_bridge2} equals
	\[ \pla\big(\conn{u_{m-1}}{u_m}{\xi(\{u_{m-1}\} \cup \thinning{\eta^{u_m}}{B}}) \big) - \tlam(u_m-u_{m-1})= -\pla\big(E(u_{m-1}, u_m;B ,\xi^{u_{m-1},u_m})\big), \]
proving~\eqref{eq:links:pi_m_proof_goal}. Lemma~\ref{lem:links:Pi_pi_identity} can now be proven by iteratively applying~\eqref{eq:links:pi_m_proof_goal}.
\end{proof}

\begin{proof}[Proof of Lemma~\ref{lem:links:oze_pis}]
For $m \in\N_0$, we can write
	\eqq{ \tlam(x_1,x_2) = \sum_{l=0}^{m} \pi_\lambda^{(l)}(x_1,x_2) +  \sum_{n\geq 0} \frac{\lambda^n}{n!} \int_{(\Rd)^n}
					\sum_{G \in \mathcal A} \weight^\pm(G) \dd \vec x_{[3,n+2]}, \label{eq:links:tlam_pis_partition}}
where $\mathcal A$ is the set of graphs $G \in \tilde{\mathcal C}_{x_1, x_2}^\pm(\vec x_{[n+2]})\setminus \tilde{\mathcal D}_{x_1, x_2}^\pm(\vec x_{[n+2]})$ together with the graphs $G \in  \tilde{\mathcal D}_{x_1, x_2}^\pm(\vec x_{[n+2]}) $ where $ \| L_{\text{min}}\| >m$. Note that if $G\in \mathcal A$, then $\textsf{Piv}^+ (x_1,x_2;G) \neq\varnothing$.

For $G \in\mathcal A$ and $u\in\textsf{Piv}^+(x_1,x_2;G)$, define
	\[ V^{\preccurlyeq}(u) := \{ y \in V(G^+): y \preccurlyeq u \} \cup \{ y \in V(G) \setminus V(G^+): \exists z\in N(y)\cap V(G^+)\textit{ with } z \prec u \},\]
that is, all the core vertices of order at most that of $u$ together with the shell vertices adjacent to at least one vertex of strictly smaller order than $u$. Next, let $u^{\text{cut}}=u^{\text{cut}}(x_1,x_2;G)$ be the vertex in $\textsf{Piv}^+(x_1,x_2;G)$ such that
	\[E\big(V^{\preccurlyeq}(u^{\text{cut}}) \setminus \{u^{\text{cut}}\}, V \setminus V^{\preccurlyeq}(u^{\text{cut}}) \big) = \varnothing
		\qquad \text{and} \qquad G[V^{\preccurlyeq}(u^{\text{cut}})] \in \tilde{\mathcal D}_{x_1,u^{\text{cut}}}^\pm .\]
If such a point exists, it is unique; if no such point exists, set $u^{\text{cut}} = x_2$. We can now partition $\mathcal A$ as
	\[ \mathcal A = \Big( \bigcup_{i=1}^{m} \mathcal A_i \Big) \cup \mathcal A_{>m}, \]
where
	\al{\mathcal A_i &:= \big\{ G \in \mathcal A: u^{\text{cut}}\neq x_2\ \text{ and } \| L_{\text{min}} (x_1, u^{\text{cut}}; G[V^{\preccurlyeq}(u^{\text{cut}})]) \| = i \big\}, \\
		\mathcal A_{>m} &:= \big\{G \in \mathcal A: \| L_{\text{min}} (x_1, u^{\text{cut}}; G[V^{\preccurlyeq}(u^{\text{cut}})]) \| >m \big\}. }
Now, if $x_s = u^{\text{cut}}$ and $V' :=V^{\preccurlyeq}(u^{\text{cut}})$ as well as $V'' := \{x_s\} \cup ( \vec x_{[n+2]} \setminus V')$, then
	\[ \weight^\pm(G) = \weight^\pm (G[V']) \weight^\pm(G[V'']), \]
that is, the weight factorizes. Therefore, for every $i\in[m]$,
	\eqq{  \sum_{n\geq 0} \frac{\lambda^n}{n!} \int_{(\Rd)^n} \sum_{G \in \mathcal A_i} \weight^\pm(G) \dd \vec x_{[3,n+2]}
				= \lambda \int_{\Rd} \pi_\lambda^{(i)}(x_1,u) \tlam(u,x_2) \dd u. \label{eq:links:tlam_pis_partition2}}
Setting $\bar R_{\lambda,m}(x_2-x_1) :=  \sum_{n\geq 0} \frac{\lambda^n}{n!} \int_{(\Rd)^n} \sum_{G \in \mathcal A_{>m}} \weight^\pm(G) \dd \vec x_{[3,n+2]}$, we can rewrite~\eqref{eq:links:tlam_pis_partition} as
	\[ \tlam(x) = \pi_{\lambda,m}(x) + \lambda \big( \pi_{\lambda,m} \ast \tlam \big)(x) + \bar R_{\lambda,m}(x).\]
One can now prove by hand or by employing Lemma~\ref{lem:links:Pi_pi_identity} that $\bar R_{\lambda,m} = R_{\lambda,m}$.
\end{proof}

\subsection{Other percolation models} \label{sec:bondperco}

The results of this paper should apply in quite analogous fashion to all other percolation models that enjoy sufficient independence---in particular, to (long-range) bond and site percolation on $\mathbb Z^d$. We pick bond percolation on $\Zd$ with edge parameter $p$ as an example. We can adjust our notation by using $\mathcal C(x,y,\Zd)$ to denote the connected subgraphs of $\Zd$ containing $x$ and $y$, and we define $\mathcal D_{x,y}(\Zd)$ and the notions for ($\pm$)-graphs analogously. Then one can show that, if we restrict to a finite box $\Lambda \subset \Zd$, the two-point function satisfies
	\eqq{ \tau_p^\Lambda(x_1, x_2) = \sum_{n \geq 0} p^n \sum_{\substack{G \in \mathcal C^\pm(x_1, x_2,\Lambda): |E(G)|=n,\\ x_1 \overset{+}{\longleftrightarrow} x_2}} (-1)^{|E^-(G)|}. 
					\label{eq:bondperco_taup}}
One can easily observe that all graphs summed over in~\eqref{eq:bondperco_taup} that contain more than one $(+)$-cluster cancel out, which is also what happens in the RCM. The direct-connectedness function can be defined analogously to Definition~\ref{def:dcfshell:shell_dcf}, providing a suitable setup for an analysis analogous to the one in Section~\ref{sec:dcfshell}.\\
\medskip

\section*{Acknowledgments} 
\noindent We thank David Brydges, Tyler Helmuth, and Markus Heydenreich for interesting discussions.

\bibliographystyle{APT}
\bibliography{OZEbib2}

\end{document}